\documentclass[11pt]{article}

\usepackage[utf8]{inputenc}

\usepackage[margin=1in]{geometry}

\usepackage{graphicx}
\usepackage{xcolor}
\usepackage{float}
\usepackage{booktabs}

\usepackage{amsmath}
\usepackage{amsthm}
\usepackage{bm}
\usepackage{newtxtext,newtxmath}   
\renewcommand{\epsilon}{\varepsilon}  

\usepackage[open,openlevel=1]{bookmark}

\usepackage{natbib}
\usepackage{authblk}
\usepackage{hyperref}

\theoremstyle{plain}
\newtheorem{theorem}{Theorem}
\newtheorem{lemma}{Lemma}

\newtheorem{proposition}{Proposition}

\newtheorem*{citedtheorem}{Theorem}  
\newtheorem*{theoremA}{Theorem 1}
\newtheorem*{theoremB}{Theorem 2}
\newtheorem*{propositionA}{Proposition 1}

\theoremstyle{definition}
\newtheorem{definition}{Definition}

\theoremstyle{remark}
\newtheorem{remark}{Remark}

\title{Differential Privacy Guarantees in Small Area Estimation}

\author[1,4]{Soumojit Das\thanks{Corresponding author. This work started when the first author was a doctoral student at the University of Maryland.}}
\author[2,3,4]{Jörg Drechsler}

\affil[1]{Washington State University, Pullman, WA, USA. \href{mailto:soumojit.das@wsu.edu}{soumojit.das@wsu.edu}}
\affil[2]{Institute for Employment Research, Germany.}
\affil[3]{Ludwig-Maximilians-Universität München, Germany.}
\affil[4]{Joint Program in Survey Methodology, University of Maryland, USA.}

\date{}

\begin{document}

\maketitle

\begin{abstract}
Statistical agencies increasingly rely on small area estimation to produce reliable estimates for subpopulations with limited sample sizes. These estimates are built from individual survey responses, so agencies must ensure that releasing them does not reveal information about any single respondent. We show that when a single draw from the posterior distribution of the Bayesian Fay-Herriot model is released, pure $\varepsilon$-differential privacy is unattainable, but the release satisfies formal privacy guarantees under R\'enyi differential privacy and zero-concentrated differential privacy without any noise being added, provided we treat the variance components as fixed. The key insight is that the posterior draw equals the posterior mean plus the Gaussian noise whose variance equals the posterior variance. The guarantee is thus governed by the sensitivity of the direct survey estimate and the posterior variance, and applies equally to a release of the posterior mean with that amount of noise added. For binary outcomes estimated with the H\'ajek estimator, the sensitivity equals the largest survey weight in the area divided by the sum of the weights. For the intercept-only model we derive exact coefficients describing how a change in one record propagates to every area's posterior mean, giving finite-sample per-area guarantees and a joint guarantee for releasing all areas at once that exceeds the largest per-area guarantee by at most a few percent in our applications. Two applications, poverty prevalence across 2,462 Public Use Microdata Areas in the American Community Survey and smoking prevalence across 52 substrata in the Washington state Behavioral Risk Factor Surveillance System, show that the guarantee is driven far more by the inequality of the survey weights than by the sample size, and that the shrinkage of the model tightens it substantially.
\end{abstract}


\section*{Statement of Significance}
Government statistical agencies routinely publish estimates for small geographic areas, such as local poverty or smoking rates. Because these estimates are built from individual survey responses, agencies must ensure that publishing them does not reveal information about any single respondent. Establishing such protection usually means adding random noise to the numbers, which lowers their accuracy and requires difficult case-by-case review before release.

This paper shows that a popular statistical method for producing small-area estimates already provides a quantifiable amount of this protection without adding any external noise. The protection comes from the randomness inherent in the method when its estimates are released as random draws rather than fixed numbers. We derive a formula for how strong the protection is and show that it depends mainly on how uneven a survey's sampling weights are. We also show that, for a widely used instantiation of the method, publishing the estimates for all areas at once provides essentially the same protection as publishing the single most exposed area alone.

Applying the approach to national surveys of poverty and smoking, we give agencies a transparent way to measure and report the privacy protection their small-area estimates already carry, and to identify design choices that would strengthen it.

\section{Introduction}
With the ever increasing amounts of data being collected and shared and increasing computing powers that can be used for sophisticated reidentification attacks, statistical agencies around the world are getting more and more concerned about ensuring sufficient protection of the confidentiality of their survey respondents when disseminating data to the public. As research conducted at the U.S. Census Bureau using data from the 2010 Decennial Census has demonstrated \citep{abowd20232010}, traditional protection strategies such as data swapping, which exchanges the values of some of the units in the data, no longer offer a sufficient level of protection.

One of the biggest challenges statistical agencies are facing when assessing the risk of their data release strategies is the fact that data are published in various forms using different channels making it difficult to track the combined risks of these data releases. Most data are released in the form of tables containing aggregate statistics. Beyond the tables, highly redacted micro datasets are often disseminated as public use files. Finally, researchers can often analyze less protected versions of the microdata at research data centers hosted by the statistical agencies. Results based on these analyses can only leave the research data center (for example, to be used in publications) if the output has been cleared by staff from the agency. The risk assessment for this dissemination channel is especially challenging, as the staff needs to consider all previous data releases to assess the incremental risk from publishing the research output. With traditional statistical disclosure limitation approaches a formal assessment of the accumulated risk of all these data releases is infeasible and statistical agencies rely on simple heuristics such as requiring that each output needs to be based on a minimum number of records. As the database reconstruction theorem \citep{dinur2003revealing} mathematically proves, such simple heuristics are not adequate to ensure low risks of disclosure. The theorem essentially states that it is  mathematically impossible to release too many statistics too accurately without risking that the underlying microdata could be fully reconstructed from the aggregate statistics.

One possible strategy to address this dilemma is to rely on the concept of differential privacy \citep{dwork2006calibrating} developed in Computer Science to ensure rigorous protection standards. Differential privacy (DP) offers formal, mathematically quantifiable privacy guarantees by bounding the effect that any single record in the data can have on the observed output.
Beyond the strong levels of protection, DP offers two important properties: immunity to postprocessing and (sequential) composition. Immunity to postprocessing implies that any function of a differentially private result also satisfies DP. Sequential composition allows quantifying the overall privacy loss over multiple data releases. DP guarantees that the overall privacy loss as measured by the privacy parameters is bounded by the sum of the individual losses from each release (see Section \ref{sec:DP_properties} for further details). This property helps to address the dilemma discussed above. If all previous data releases satisfy DP, it will be easy for the staff working at the research data center to assess the additional privacy leakage from the research output the analyst would like to use outside the research data center (assuming the output also satisfies DP).

However, implementing DP in practice can be challenging. Since DP bounds the risk under worst case assumptions, large amounts of noise often need to be added to achieve meaningful privacy guarantees under DP especially for statistics that are sensitive to small changes in the underlying data. For example, \cite{barrientos2024feasibility} find that none of the existing methods to achieve DP for linear regression models achieve utility levels that would be accepted by applied researchers based on meaningful levels of the privacy parameters. The situation is even more difficult in the survey context as the impacts of sampling, weighting, linkage, and imputation all need to be taken into account when trying to establish formal privacy guarantees \citep{bun2022controlling,Das2022,drechsler2023differential,lin2024differentiallyci,drechsler2024complexities, lin2024differentially}.

For these reasons there is growing interest in recent years in estimation techniques that inherently provide privacy guarantees because they have some randomness baked into the estimation process, which can be linked to the formal guarantees of differential privacy. Specifically, starting with \cite{mir2013differential}, several papers have investigated which formal privacy guarantees can be established for Bayesian inference obtained by sampling from the posterior distribution \citep{wang2015privacy,minami2016differential,zhang2016differential,jewson2023differentially,zhang2023dp}.

Our paper adds to this discussion by investigating which formal guarantees Bayesian small area estimation strategies can provide. The term small area estimation (SAE) subsumes a class of estimation techniques that aim to improve inference for sparsely populated subgroups of the data by pooling the information across subgroups through hierarchical models with the goal of reducing the uncertainty in the obtained estimates. The U.S. Census Bureau and its various research data centers release model-based estimates for small areas which are subject to disclosure review \citep{fsrdc_handbook}, and yet to the best of our knowledge, no formal quantification of the disclosure risk of these outputs is available. Being able to establish formal guarantees for small area models would give such reviews a principled basis. We aim to achieve this for one specific SAE model: the Bayesian version of the classical and still widely popular Fay-Herriot Model \citep{fay1979estimates}.

The remainder of the paper is organized as follows. Section~\ref{sec:background} reviews the necessary background on differential privacy, Bayesian posterior sampling, and small area estimation. Section~\ref{sec:privacy-guarantees} presents the theoretical results: a general privacy guarantee for an arbitrary design matrix (Theorem~\ref{thm:rdp}), exact finite-$m$ guarantees for the intercept-only model under area-specific variances (Proposition~\ref{prop:intercept-varying}), and an analysis of how the bound depends on survey design parameters.  Section~\ref{sec:applications} illustrates the framework empirically using American Community Survey microdata and Behavioral Risk Factor Surveillance System data. Section~\ref{sec:discussion} discusses limitations and broader applicability. Proofs and additional results are provided in the supplement.

\section{Background}\label{sec:background}

\subsection{Differential Privacy}\label{sec:DP_prelim}

Differential privacy (DP) provides a mathematically rigorous framework for quantifying privacy guarantees by bounding the effect that any single record can have on the observed output~\citep{dwork2006calibrating}. Several variants of DP have been proposed in the literature in recent years (\citet{desfontaines2019sok} identifies more than 200). We will only review the variants relevant for our paper: $\varepsilon$-DP, $(\varepsilon,\delta)$-DP \citep{dwork2006calibrating}, R\'enyi-DP \citep{mironov2017renyi}, and $\rho$-Zero-concentrated-DP \citep{bun2016concentrated}.

\begin{definition}[Neighboring Databases]
Two databases $D$ and $D'$ are considered neighbors (denoted $D \sim D'$) if they differ in exactly one record.
\end{definition}
Two types of neighboring datasets are distinguished in the DP literature. Under unbounded DP two datasets are called neighbors if $D$ can be obtained from $D'$ by adding or removing a single record. Under bounded DP the size of the data is fixed and $D$ can be obtained from $D'$ by changing the values of a single record in the database. We only consider the bounded case in this paper as it allows the survey weights to be treated as fixed between neighboring datasets under certain circumstances. This helps to substantially reduce the amount of uncertainty that is required in the obtained estimates to satisfy differential privacy. See \citet{drechsler2023differential} and \citet{drechsler2024complexities} for further discussions of the complexities that arise when considering DP in the survey context.

\subsubsection{DP Definitions}
\begin{definition}[Pure $\epsilon$-Differential Privacy \citep{dwork2006calibrating}]
A randomized algorithm $f: \mathcal{D} \mapsto \mathcal{R}$ satisfies $\epsilon$-differential privacy if, for any neighboring databases $D, D' \in \mathcal{D}$ and any measurable set $S \subseteq \mathcal{R}$:
\begin{equation}
\Pr[f(D) \in S] \leq e^{\epsilon} \Pr[f(D') \in S].
\end{equation}
\end{definition}

The parameter $\varepsilon$ is called the privacy loss parameter. Larger values of $\varepsilon$ allow for larger differences in the output distribution if a single record is changed in the input database implying weaker levels of protection.

\begin{definition}[Approximate $(\epsilon, \delta)$-Differential Privacy \citep{dwork2006calibrating}]
A randomized algorithm $f: \mathcal{D} \mapsto \mathcal{R}$ satisfies $(\epsilon, \delta)$-differential privacy if, for any neighboring databases $D, D' \in \mathcal{D}$ and any measurable set $S \subseteq \mathcal{R}$:
\begin{equation}
\Pr[f(D) \in S] \leq e^{\epsilon} \Pr[f(D') \in S] + \delta.
\end{equation}
\end{definition}

\begin{definition}[R\'enyi Divergence]
For two probability distributions $P$ and $Q$, the R\'enyi divergence of order $\alpha > 1$ is:
\begin{equation}
D_{\alpha}(P \|Q) = \frac{1}{\alpha - 1} \ln\int_x P(x)^{\alpha} Q(x)^{1-\alpha} dx = \frac{1}{\alpha - 1} \ln E_Q\left[\left(\frac{P(x)}{Q(x)}\right)^{\alpha}\right]
\end{equation}
\end{definition}

\begin{definition}[$(\alpha, \epsilon)$-R\'enyi Differential Privacy \citep{mironov2017renyi}]
A randomized mechanism $f: \mathcal{D} \mapsto \mathcal{R}$ satisfies $(\alpha, \epsilon)$-R\'enyi differential privacy (RDP) if for any neighboring databases $D, D' \in \mathcal{D}$:
\begin{equation}
D_{\alpha}(f(D) \| f(D')) \leq \epsilon.
\end{equation}
\end{definition}

R\'enyi differential privacy provides a relaxation of pure differential privacy that often yields tighter composition bounds, that is, tighter bounds regarding how the privacy loss accumulates over multiple data releases. The privacy parameters of R\'enyi DP can be converted to $(\epsilon, \delta)$-DP guarantees~\citep{mironov2017renyi}, which can be helpful when trying to interpret the privacy guarantees. Specifically, $(\alpha, \epsilon)$-RDP implies $(\epsilon + \frac{\log(1/\delta)}{\alpha-1}, \delta)$-DP for any $0 < \delta < 1$.

\begin{definition}[$\rho$-Zero Concentrated Differential Privacy \citep{bun2016concentrated}]
A randomized mechanism $f: \mathcal{D} \mapsto \mathcal{R}$ satisfies $\rho$-zero concentrated differential privacy ($\rho$-zCDP) if for all $\alpha > 1$ and adjacent databases $D, D' \in \mathcal{D}$:
\begin{equation}
D_{\alpha}(f(D) \| f(D')) \leq \rho \alpha.
\end{equation}
\end{definition}

The $\rho$-zCDP definition is closely related to R\'enyi-DP but requires only one parameter, often yielding slightly stronger results and simpler analysis~\citep{bun2016concentrated}. The $\rho$-zCDP guarantee can be converted to $(\varepsilon, \delta)$-DP using the closed-form expression~\citep{bun2016concentrated}:
\begin{equation}\label{eq:zcdp-conversion}
\varepsilon = \rho + 2\sqrt{\rho \log(1/\delta)}
\end{equation}
for any $0 < \delta < 1$. Since $\delta$ is a free parameter of this conversion, we report $\rho$ as the primary privacy measure throughout the paper and give $\varepsilon$ at a stated $\delta$ only as a secondary summary.

\subsubsection{The Gaussian Mechanism}

\begin{definition}[Global $\ell_2$-Sensitivity]
For a function $f: \mathcal{D} \to \mathbb{R}^k$, the global $\ell_2$-sensitivity is:
\begin{equation}
\Delta_2^f = \max_{D \sim D'} \|f(D) - f(D')\|_2
\end{equation}
where the maximum is over all pairs of neighboring databases.
\end{definition}

\begin{definition}[Gaussian Mechanism \citep{dwork2014algorithmic}]
For a function $f: \mathcal{D} \to \mathbb{R}^k$ with $\ell_2$-sensitivity $\Delta_2^f$, the Gaussian mechanism releases:
\begin{equation}
\mathcal{M}(D) = f(D) + Z, \quad Z \sim N(0, \sigma^2 I_k)
\end{equation}
where $\sigma^2 = (\Delta_2^f)^2 / (2\rho)$ for $\rho$-zCDP, or equivalently $\sigma^2 = 2\log(1.25/\delta) \cdot (\Delta_2^f)^2 / \varepsilon^2$ for $(\varepsilon, \delta)$-DP.
\end{definition}

The Gaussian mechanism satisfies $(\alpha, \alpha(\Delta_2^f)^2/(2\sigma^2))$-R\'enyi DP for every $\alpha > 1$ \citep{mironov2017renyi} and, since this bound is linear in $\alpha$, $\rho$-zCDP with $\rho = (\Delta_2^f)^2 / (2\sigma^2)$ \citep{bun2016concentrated}.
This relationship will be central to our analysis: in Section~\ref{sec:privacy-guarantees}, we show that the Bayesian Fay-Herriot posterior draw is equivalent to a Gaussian mechanism applied to the posterior mean, with the posterior variance playing the role of $\sigma^2$.

\subsubsection{Key Properties of Differential Privacy}\label{sec:DP_properties}

All variants of differential privacy discussed in the previous section offer several attractive properties that help monitor the privacy loss in practical settings. We only review the two properties that will be relevant for the remainder of this paper and present them in the context of $\varepsilon$-DP. Slightly modified versions of these properties can also be specified for the various DP relaxations discussed above.

\noindent \textbf{Post-processing Immunity}: If a mechanism $f$ satisfies $\epsilon$-DP, then for any function $g$, the composed mechanism $g \circ f$ also satisfies $\epsilon$-DP.

\noindent \textbf{Sequential Composition}: For mechanisms $f_1, \ldots, f_k$ satisfying $\epsilon_1$-DP, \ldots, $\epsilon_k$-DP respectively, their overall privacy loss is bounded by $(\sum_{i=1}^k \epsilon_i)$-DP.

\subsection{Differential Privacy and Posterior Sampling}\label{sec:OPS}

A key ingredient of DP is the privacy loss random variable (PLRV). For a randomized mechanism $f$ and a specific output $\theta$, the PLRV is defined as the log-likelihood ratio:

$$\mathcal{L}(\theta) = \log \frac{f_D(\theta)}{f_{D'}(\theta)}.$$
Pure $\varepsilon$-DP requires that this quantity is bounded by $\varepsilon$ for all possible outputs.

As pure $\varepsilon$-DP requires a strict upper bound for this ratio, no deterministic mechanism/analysis model (including the classical Fay-Herriot model) can satisfy $\varepsilon$-DP (unless changing an arbitrary record does not change the output of the mechanism/model).

However, in Bayesian Statistics inferences are commonly based on random draws from the posterior distribution of the parameters given the data. This implies that the outcome of interest is a random variable and thus there is a possibility that the PLRV can be bounded. Thus, the natural question arises under which conditions these posterior draws satisfy the requirements of DP and thus might offer privacy ``for free''.

The most general result is given by \citet{wang2015privacy} based on earlier work by \citet{dimitrakakis2014robust}. Denoting $\ell(X \mid \theta)$ to be the likelihood, the authors show that formal privacy guarantees can always be achieved as long as the likelihood is bounded:
\begin{citedtheorem}[\citet{wang2015privacy}]
    If $\mathrm{sup}_{X \in \mathcal{X}, \theta \in \Theta} |\log(\ell(X \mid \theta))| \leq B$ then one sample from $\theta \mid X$ with any prior preserves $(4B, 0)$-DP.
\end{citedtheorem}
Unfortunately, the likelihood of the Fay-Herriot model is unbounded and thus these results are not helpful for our endeavor. \citet{minami2016differential} derived results for the more general case of the Gibbs posterior (of which the Bayesian posterior is a special case) that do not require bounds for the loss function to achieve $(\varepsilon,\delta)$-DP. The remaining literature either focuses on results for specific analysis tasks or on modifying the prior to ensure that draws from the posterior satisfy some formal privacy guarantees. The idea of modifying the prior to achieve (approximate)-DP for draws from the posterior was adopted in various papers \citep{dimitrakakis2017differential,zhang2016differential} and most extensively in the thesis by \citet{Zhengdifferential}.
Since all these strategies modify the posterior sampling approach in some way, none of them is suitable for our goal of measuring the inherent privacy guarantees of the Fay-Herriot model.

\subsection{Small Area Estimation}\label{sec:SAE}

Small area estimation (SAE) addresses the challenge of producing reliable estimates for subpopulations with small sample sizes by borrowing strength across areas through hierarchical modeling~\citep{rao2015small}. SAE methods can be broadly classified into two modeling strategies:

\begin{enumerate}
\item \textbf{Area-level models} use auxiliary information only on the area-level, for example, average age or unemployment rate at the area level \citep{fay1979estimates, das2025multidimensional}.
\item \textbf{Unit-level models} utilize individual-level auxiliary information \citep{battese1988error, das2025approximate}.
\end{enumerate}

Our focus is on area-level models, which have the practical advantage of requiring only aggregate auxiliary information.

\subsubsection{The Fay-Herriot Model}
\label{sec:Fay-Herriot}
The Fay-Herriot model~\citep{fay1979estimates} remains one of the most widely used area-level SAE models. Let $\theta_i$ be the population parameter of interest for area $i$, $i = 1, \ldots, m$, and let $y_i$ be the direct survey estimate of $\theta_i$ with known sampling variance $\sigma_{y_i}^2$.
The model is specified hierarchically as:

 \textbf{Level 1 (Sampling model)}:
\begin{equation}
y_i \mid \theta_i \stackrel{ind}{\sim} N(\theta_i, \sigma_{y_i}^2)
\end{equation}

\textbf{Level 2 (Linking model)}:
\begin{equation}
\theta_i \stackrel{ind}{\sim} N(x_i'\beta, \sigma^2_v)
\end{equation} where $x_i$ is a vector of known auxiliary variables for area $i$, $\beta$ is a vector of unknown regression coefficients, and $\sigma^2_v$ is an unknown model variance component.

In most practical applications, the direct estimate $y_i$ is a survey-weighted estimator for the population mean in area $i$. Two types of estimators are commonly used in practice: The Horvitz-Thompson (HT) estimator and the H\'ajek estimator \citep{lohr2021sampling}. The HT estimator is given as
\begin{equation}\label{eq:HT}
\hat{y}_i^{HT} = \frac{\sum_{r=1}^{n_i} w_{ir}\,Y_{ir}}{N_i}
\end{equation}
where $r$ indexes the $n_i$ sampled units in area $i$, $w_{ir}$ is the survey weight for unit $r$, $Y_{ir}$ is the outcome value and $N_i$ is the total number of units in area $i$ in the population.
For the H\'ajek estimator, the denominator $N_i$ in Eq. \ref{eq:HT} is replaced with its estimate $\sum_{r=1}^{n_i} w_{ir}$:
\begin{equation}\label{eq:hajek}
\hat{y}_i^{HJ} = \frac{\sum_{r=1}^{n_i} w_{ir}\,Y_{ir}}{\sum_{r=1}^{n_i} w_{ir}}
\end{equation}

Unless otherwise noted, we will focus on the H\'{a}jek estimator in the remainder of this paper since this estimator is more commonly used in practice (see \citet{lohr2021sampling} for further discussion on this topic).

While the Fay-Herriot model assumes normality for the direct survey estimates $y_i$, it has been successfully applied with binary outcomes in numerous small area estimation contexts (\cite{liu2007hierarchical}; \cite{esteban2012small}; \cite{benavent2016multivariate}; \cite{krenzke2020hierarchical}). This matters because the sensitivity, and with it the privacy guarantee, scales with the range of the outcome (Theorem~\ref{thm:rdp} and section~\ref{sec:analytical}), which equals one for a binary outcome.

\subsubsection{Estimation under the Fay-Herriot Model}

The Best Predictor (BP) of $\theta_i$ when parameters are known is:
\begin{equation}
\theta_i^{BP} = (1 - B_i)y_i + B_i x_i^T\beta
\end{equation} where $B_i = \sigma_{y_i}^2/(\sigma_{y_i}^2 + \sigma^2_v)$ is the shrinkage factor.

The Best Linear Unbiased Predictor (BLUP) is obtained by replacing $\beta$ with its weighted least squares (WLS) estimate:
\begin{equation}
\theta_i^{BLUP} = (1 - B_i)y_i + B_i x_i^T\hat{\beta}_{WLS}
\end{equation} where
\begin{equation}
\hat{\beta}_{WLS} = \left[\sum_{j=1}^{m} (1-B_j)x_j x_j^T \right]^{-1} \sum_{j=1}^{m} (1-B_j)x_j y_j
\end{equation}

\subsubsection{Bayesian Fay-Herriot Model}

Under the Bayesian paradigm, we place priors on the unknown parameters. For our analysis, we assume a flat prior on $\beta$: $\pi(\beta) \propto 1$. We also assume that the sampling variance $\sigma_{y_i}^2$ and model variance $\sigma_v^2$ are known to simplify the derivations presented in the next section. We come back to this point in Section~\ref{sec:robustness}.
The posterior distribution of $\theta_i$, which forms the basis of our privacy analysis, is:
\begin{equation}
\label{eqn: posterior}
    \theta_i \mid D, \sigma_v^2 \sim N\left(\frac{\sigma_v^2 y_i + \sigma_{y_i}^2 x_i^T\hat{\beta}_{WLS}}{\sigma_v^2 + \sigma_{y_i}^2}, \frac{\sigma_{y_i}^2 \sigma_v^2}{\sigma_{y_i}^2 + \sigma_v^2} + B_i^2 \sigma_v^2 \cdot x_i^T\left[\sum_{j=1}^{m} (1-B_j)x_j x_j^T \right]^{-1} x_i\right).
\end{equation}
Under the flat prior assumption, the posterior mean has a closed-form expression identical to the empirical best linear unbiased predictor (the BLUP (defined above) with $\sigma_v^2$ replaced by its estimate $\hat{\sigma}_v^2$). However, as discussed above, deterministic outputs cannot satisfy differential privacy without additional noise injection, as the privacy loss random variable would be unbounded. Therefore, building on the one posterior sampling literature of Section~\ref{sec:OPS}, our analysis instead assumes release of a single posterior draw $\tilde{\theta}_i \sim N(\mu_i^D, \sigma_i^2)$ with
\[
\mu_i^D=\frac{\sigma_v^2 y_i + \sigma_{y_i}^2 x_i^T\hat{\beta}_{WLS}}{\sigma_v^2 + \sigma_{y_i}^2}\quad\text{and}\quad \sigma_i^2=\frac{\sigma_{y_i}^2 \sigma_v^2}{\sigma_{y_i}^2 + \sigma_v^2} + B_i^2 \sigma_v^2 \cdot x_i^T\left[\sum_{j=1}^{m} (1-B_j)x_j x_j^T \right]^{-1} x_i.
\]

For our theoretical development, we simplify by assuming equal sampling variances across areas: $\sigma_{y_i}^{2} = \sigma_y^2$ for all $i$.
\begin{equation}\label{eq:sigma_i}
\sigma_i^2 = \frac{\sigma_y^2 \sigma_v^2}{\sigma_y^2 + \sigma_v^2} + B^2 \sigma_v^2 \cdot x_i^T\left[\sum_{j=1}^{m} (1-B)x_j x_j^T \right]^{-1} x_i
\end{equation}
and $B = \sigma_y^2/(\sigma_y^2 + \sigma_v^2)$. Note that the posterior variance $\sigma_i^2$ is area-specific through its dependence on $x_i$.

The assumption of constant variances may not always be justified. In Section~\ref{sec:intercept-exact}, we present results that do not require this assumption for the special case of the intercept-only model, which we also use in our applications in Section~\ref{sec:applications}.
Extending the formal guarantees to a general design matrix with varying variances is left for future work.

\section{Privacy Guarantees}\label{sec:privacy-guarantees}

We require the following assumptions to establish our results:
\begin{enumerate}
    \item The variance components $\sigma^2_y$ and $\sigma^2_v$ are known (or can be treated as fixed).
    \item The area-level auxiliary information $x_i$ can be considered public knowledge.
\end{enumerate}

The second assumption is relatively mild, as $x_i$ typically comes from external data sources such as publicly released Census tables. The first assumption is more restrictive but standard in the SAE literature. From an inferential perspective, the effect of using plug-in estimators is typically small. However, in the DP context, estimating parameters from the data without accounting for their privacy leakage would invalidate formal guarantees. We will come back to this issue in Section~\ref{sec:robustness}.

\subsection{Setup and Mechanism Description}

Consider two neighboring databases $D$ and $D'$ that differ only in one record in one area. Without loss of generality, we assume the two datasets differ in area $k$. Formally under bounded DP, we have:
\begin{itemize}
\item $D = \{(x_i, y_i): i = 1, \ldots, m\}$
\item $D' = \{(x_j, y_j): j = 1, \ldots, m\}$
\item $y_i^{(D)} = y_j^{(D')}$ for all $i = j \neq k$
\item $x_k^{(D)} = x_k^{(D')}$ (area-level covariates are public)
\item $y_k^{(D)} \neq y_k^{(D')}$
\end{itemize}

Let $\Theta_i$ denote the mechanism that outputs an estimate for area $i$ by drawing once from the posterior distribution:
\begin{equation}
\tilde{\theta}_i = \Theta_i(D) \sim N(\mu_i^D, \sigma_i^2)
\end{equation}
where $\mu_i^D= (\sigma_v^2 y_i^{(D)} + \sigma_y^2 x_i^T\hat{\beta}_{WLS}^{(D)})/(\sigma_v^2 + \sigma_y^2)$ is the posterior mean under database $D$ and $\sigma_i^2$ is the posterior variance defined above (which depends on $x_i$ but not on the data values).

Our goal is to characterize how the distribution of $\tilde{\theta}_i$ changes when the database changes from $D$ to $D'$ as measured by the PLRV $\mathcal{L}(\tilde{\theta}_i)$.

\subsection{Why \texorpdfstring{$\varepsilon$}{epsilon}-DP Cannot Be Achieved}\label{sec:no_eps_dp}

\begin{theorem}[Privacy Loss for the Bayesian Fay-Herriot Model]\label{thm:expected-privacy-loss}
Let $\Theta_i$ be the mechanism that outputs an estimate for area $i$ by drawing once from the posterior distribution given in Equation~(\ref{eqn: posterior}). Assume the variance parameters $\sigma^2_v$ and $\sigma_{y}^2$ are known, and the area-level information $x_i$ is public. For any two neighboring databases $D$ and $D'$ that differ in one record in area $k$, the privacy loss random variable $\mathcal{L}(\tilde{\theta}_i)$ under $\tilde{\theta}_i \sim N(\mu_i^D, \sigma_i^2)$ follows a Gaussian distribution:
    \begin{equation}\label{eq:privacy-loss-distribution}
    \mathcal{L}(\tilde{\theta}_i) \sim N\left(\frac{(\mu_i^D - \mu_i^{D'})^2}{2\sigma_i^2}, \frac{(\mu_i^D - \mu_i^{D'})^2}{\sigma_i^2}\right)
    \end{equation}
\end{theorem}

\begin{proof}
See Section~S1.1 of the supplement.
\end{proof}

Since $\epsilon$-DP requires $\mathcal{L}(\Theta)$ to be bounded in $[-\varepsilon,\varepsilon]$ but $\mathcal{L}(\tilde{\theta}_i)$ has unbounded support, the mechanism $\Theta_i$ can never satisfy $\varepsilon$-DP.

\subsection{Connection to the Gaussian Mechanism and R\'enyi DP}\label{sec:rdp}
While achieving pure $\varepsilon$-DP is impossible for the Bayesian Fay-Herriot model, obtaining results for R\'enyi DP and $\rho$-zCDP is straightforward once we notice the close connection between the mechanism $\Theta_i$ and the Gaussian mechanism. The mechanism $\Theta_i$ can be interpreted as a Gaussian mechanism that adds zero-centered noise to the posterior mean. Thus, to derive the privacy guarantees of $\Theta_i$, we need to set $\sigma_i^2$ equal to the variance of the Gaussian mechanism ($(\Delta^f_2)^2\alpha/(2\varepsilon)$ for R\'enyi DP or $(\Delta^f_2)^2/(2\rho)$ for $\rho$-zCDP) and solve for the respective privacy parameter. The major challenge lies in computing the sensitivity $\Delta$ for the Bayesian Fay-Herriot Model, which we derive in Section~S1.2 of the supplement.

\begin{theorem}[Privacy Guarantees for the Bayesian Fay-Herriot Model under R\'enyi Differential Privacy]\label{thm:rdp}
Let $\Theta_i$ be the mechanism that outputs an estimate for area $i$ by drawing once from the posterior distribution given in Equation~(\ref{eqn: posterior}). Assume the variance parameters $\sigma^2_v$ and $\sigma_{y}^2$ are known, and the area-level information $x_i$ is public. For any two neighboring databases $D$ and $D'$ that differ only in one record in area $k$, the mechanism $\Theta_i$ satisfies $(\alpha, \varepsilon_\alpha)$-R\'enyi differential privacy for any $\alpha > 1$, with
\begin{equation}\label{eq:Renyi_guarantee}
    \varepsilon_i^{(\alpha)}\leq\frac{\alpha  [\max\{\Delta_{y_i};B\,S_y\} ]^2}{2\sigma_i^2}\leq\frac{\alpha S^2_{y}}{2\sigma_{i}^2}
\end{equation}
where $\Delta_{y_i}=\max_{D \sim D'} |y_i^{(D)} - y_i^{(D')}|$ is  sensitivity of the direct survey estimate for area $i$ and $S_{y} =  \max_i\Delta_{y_i}$ is the maximum sensitivity across all direct estimators. Under bounded DP and fixed weights, using the Horvitz-Thompson estimator as the direct estimator leads to $\Delta_{y_i}^{HT}=\frac{w_{\max,i}R_{y_i}}{N_i}$, where $w_{\max,i}=\max_r w_{ir}$ is the maximum weight in area $i$ and $R_{y_i}$ is the range of possible values of $Y$ in area $i$ ($R_{y_i} = 1$ for a binary outcome). The sensitivity for the H\'ajek estimator is $\Delta_{y_i}^{HJ}=\frac{w_{\max,i}R_{y_i}}{\sum_r w_{ir}}$.
\end{theorem}

\begin{proof}
See Section~S1.2 of the supplement.
\end{proof}

\begin{remark}[Sensitivity of the direct estimators]\label{rem:sensitivity}
Under bounded (substitution) DP with fixed weights, the sensitivity of both direct estimators is approximately the same: changing one respondent's outcome $Y_{kr}$ shifts the numerator of both estimators by $w_{kr}$, and neither denominator changes: the H\'{a}jek denominator $\sum_r w_{ir}$ is unaffected because the same respondents remain in the sample with the same weight, and the HT denominator $N_i$ is a fixed population total. The sensitivities are therefore approximately identical:
\begin{equation}
\Delta_{y_i}^{\,HJ} = \frac{w_{\max,i}\,R_{y_i}}{\sum_r w_{ir}} \approx \frac{w_{\max,i}\,R_{y_i}}{N_i} = \Delta_{y_i}^{\,HT}
\end{equation}
where the approximation holds because $\sum_r w_{ir}$ estimates $N_i$; the two coincide exactly only for designs in which $\sum_r w_{ir} = N_i$.

One might conjecture that the common practice of standardizing the survey weights for the H\'{a}jek estimator (e.g., rescaling so that $w_{ir}^* = w_{ir}/\bar{w}_i$) could reduce the sensitivity by shrinking $w_{\max}$. However, the H\'{a}jek estimator is a ratio of weighted sums and is therefore invariant to any multiplicative rescaling of the weights: the rescaling factor cancels in numerator and denominator simultaneously.
\end{remark}

\begin{remark}[Composition for Multiple Areas]\label{rem:composition}
Theorem~\ref{thm:rdp} provides the RDP guarantee for releasing a single area's estimate. When estimates for all $m$ areas are released, sequential RDP composition yields $(\alpha, \sum_{i=1}^m \varepsilon_i^{(\alpha)})$-RDP, which can be conservative. For the intercept-only model which we discuss in Section \ref{sec:intercept-exact}, Proposition~\ref{prop:intercept-varying} provides a tighter joint-release bound by directly evaluating the R\'enyi divergence of the joint release.
\end{remark}

\subsection{Relationship with \texorpdfstring{$\rho$}{rho}-zCDP}\label{sec:cZDP_results}

The R\'enyi DP bound has a direct connection to $\rho$-zero concentrated differential privacy ($\rho$-zCDP) \citep{bun2016concentrated}. Since the bound takes the form $\varepsilon_\alpha = \alpha \cdot c$ for some constant $c$, the mechanism satisfies $\rho$-zCDP with:
\begin{equation}\label{eq:czdp_guarantee}
\rho_i = \frac{[\max\{\Delta_{y_i};B\,S_y\}]^2}{2\sigma_{i}^2}\leq\frac{S_{y}^2}{2\sigma_{i}^2}
\end{equation}
The conversion to $(\varepsilon,\delta)$-DP is given in Equation~(\ref{eq:zcdp-conversion}).

\subsection{Exact Guarantees for the Intercept-Only Model}\label{sec:intercept-exact}

For the special case of the intercept-only Fay-Herriot model ($x_i = 1$ for all $i$) that we use in both applications in Section \ref{sec:applications}, we can obtain tight privacy bounds without the assumption of equal sampling variances across areas.

\begin{proposition}[Privacy guarantees for the intercept-only model with area specific variances]\label{prop:intercept-varying}
Consider the intercept-only Fay-Herriot model with area-specific sampling variances $\sigma_{y_i}^2$, shrinkage factors $B_i = \sigma_{y_i}^2/(\sigma_{y_i}^2 + \sigma_v^2)$, and $W_+ = \sum_{j=1}^m(1-B_j)$. Assume the variance components are publicly available and the survey weights are fixed under bounded DP. For any neighboring pair differing in one record in area $k$, the change in the posterior mean of area $i$ is exactly $\mu_i^D - \mu_i^{D'} = c_{ik}\,(y_k^{(D)} - y_k^{(D')})$, where
\begin{equation}\label{eq:exact-coefficients}
c_{ik} = \begin{cases} (1-B_k)\left(1 + \dfrac{B_k}{W_+}\right), & i = k, \\[2mm] \dfrac{B_i\,(1-B_k)}{W_+}, & i \neq k. \end{cases}
\end{equation}
Consequently:
\begin{enumerate}
\item[(i)] The mechanism $\Theta_i$ releasing a single posterior draw for area $i$ satisfies $\rho_i$-zCDP with
\begin{equation}\label{eq:rho-per-area}
\rho_i^{Int} = \frac{[\max_{1 \leq k \leq m} (c_{ik}\Delta_{y_k})]^2}{2\sigma_i^2}\leq\frac{S^2_y}{2\sigma^2_i}.
\end{equation}
\item[(ii)] The mechanism $\Theta = (\Theta_1, \ldots, \Theta_m)$ releasing draws for all $m$ areas satisfies $\rho_{\mathrm{joint}}$-zCDP with
\begin{equation}\label{eq:rho-joint}
\rho_{\mathrm{joint}} =  \max_{1 \leq k \leq m}\; \frac{\Delta_{y_k}^2}{2} \sum_{i=1}^m \frac{c_{ik}^2}{\sigma_i^2}.
\end{equation}
\end{enumerate}
The coefficients satisfy $c_{kk} \leq 1$ and $c_{ik} \leq B_i$ for $i \neq k$; these bounds hold exactly for any finite $m$.
\end{proposition}

\begin{proof}
See Section~S1.3 of the supplement.
\end{proof}

The coefficients $c_{ik}$ are identities, not bounds: the only inequality in the result is $|y_k^{(D)} - y_k^{(D')}| \leq \Delta_{y_k}$. The maximum in (\ref{eq:rho-per-area}) is over the area $k$ in which the record changes. Whether it is attained at $k = i$ can be verified by checking $c_{ik}\Delta_{y_k} \leq c_{ii}\Delta_{y_i}$ for all $k \neq i$. This holds in both applications in Section~\ref{sec:applications}, so the per-area guarantee there is $\rho_i^{Int} = c_{ii}^2\,\Delta_{y_i}^2/(2\sigma_i^2)$ with $c_{ii} = (1-B_i)(1 + B_i/W_+)$. The factor $c_{ii}^2$ measures how much tighter the exact guarantee is than the bound $\Delta_{y_i}^2/(2\sigma_i^2)$ obtained from $c_{kk} \leq 1$; for a large number of areas, $c_{ii}^2 \approx (1-B_i)^2$, so the tightening is governed by the shrinkage of area $i$.

The joint-release bound (\ref{eq:rho-joint}) takes a maximum over the perturbed area $k$ of a sum over all $m$ released areas; it is not an additive composition over areas. For each fixed $k$, the sum splits exactly into its $i = k$ term and a cross-area remainder:
\begin{equation}\label{eq:rho-joint-decomp}
\frac{\Delta_{y_k}^2}{2}\sum_{i=1}^m \frac{c_{ik}^2}{\sigma_i^2} = \frac{\Delta_{y_k}^2\, c_{kk}^2}{2\sigma_k^2}\,\left(1 + r_k\right), \qquad r_k = \frac{\sigma_k^2}{(W_+ + B_k)^2}\sum_{i \neq k}\frac{B_i^2}{\sigma_i^2}.
\end{equation}
The remainder $r_k$ aggregates $m-1$ terms of order $1/W_+^2$ and is therefore $O(m/W_+^2)$, vanishing as $O(1/m)$ whenever the average shrinkage is bounded away from one (areas that arrive essentially fully shrunk add to $m$ but not to $W_+$, and do not dilute the remainder); in the applications it is computed exactly and satisfies $r_k < 10^{-4}$ (ACS) and $r_k \leq 0.05$ (BRFSS) for every $k$ (Section~\ref{sec:applications}). Hence $\rho_{\mathrm{joint}} \approx \max_k\, \Delta_{y_k}^2\,c_{kk}^2/(2\sigma_k^2)$: releasing all $m$ areas costs essentially no more than the single most sensitive area. In particular, $\rho_{\mathrm{joint}} \geq \max_i \rho_i^{Int}$ always holds.

\subsection{Robustness to Variance Estimation}\label{sec:robustness}

Theorems~\ref{thm:expected-privacy-loss},~\ref{thm:rdp}, and Proposition~\ref{prop:intercept-varying} assume the variance components $\sigma^2_v$ and $\sigma^2_{y_i}$ are known. In practice, both $\sigma^2_v$ and $\sigma^2_{y_i}$ are estimated from the sensitive data. This violates the formal privacy guarantees as any usage of the unprotected data implies some privacy leakage that needs to be taken into account.

There are three possible options to address this problem:
\begin{enumerate}
    \item Directly quantifying the additional privacy leakage from using the two variance components.
    \item
    Obtaining differentially private estimates of the two variance components to be used as plug-in estimates in the downstream calculations.
    \item Being pragmatic by treating the quantities as known and accepting the fact that the privacy guarantees only hold, if the two quantities are treated as invariants.
    \end{enumerate}

Option 1 would require fundamental adjustments for computing the privacy leakage. The posterior could no longer be stated in closed form and MCMC methods would be required. However, we postulate that the leakage would grow substantially, as the data will be accessed at each iteration of the algorithm.

Option 2 would require privatized estimates for the two variance components. We note that the overall privacy loss when using these privatized estimates to compute the posterior means would simply be the sum over the budget spent when computing the privatized variance estimates and the calculated budget based on Equation (\ref{eq:Renyi_guarantee}) or Equation (\ref{eq:czdp_guarantee}) due to the composition and immunity to postprocessing properties of differential privacy. However, the biggest disadvantage of this option would be that we would no longer get privacy for free and applied researchers would have to learn how to compute the privatized variance estimates.

Option 3 has the obvious advantage that no further adjustments are necessary. However, since the data are accessed to estimate the two variance components, there is some privacy leakage that is not accounted for when deriving the privacy guarantees based on Eq. (\ref{eq:Renyi_guarantee}) or Eq. (\ref{eq:czdp_guarantee}). From a practical perspective, this leakage might be considered to be small--how much more can I learn about any individual, by knowing that the released small area means were computed using variance estimates computed from the data? Still, from a formal standpoint, we no longer can make any claims regarding the formal privacy guarantees of these means.

We can circumvent this dilemma by declaring the two variance components to be invariants. Invariants are aspects of the data that are declared to be fixed, i.e., they do not change over neighboring datasets.
In our context, declaring the two variance components to be invariant would imply that the domain of neighboring datasets $D$ and $D'$ would be limited to those datasets that would result in the same estimates for the two variance components.  While this is a substantial reduction in the size of the domain, the formal privacy guarantees would still hold for this domain. Whether this domain would still be large enough to provide meaningful privacy guarantees would be another interesting topic for future research.

\subsection{Analytical Properties}\label{sec:analytical}

The conservative privacy guarantee $\rho_i\leq S_y^2/(2\sigma^2_i)$ that holds for both settings discussed in Section \ref{sec:cZDP_results} and \ref{sec:intercept-exact} depends on two quantities: the maximum of the sensitivities of the direct estimators  $S_y=\max_i\Delta_{y_i}$ and the posterior variance $\sigma_i^2$.
For the H\'{a}jek estimator under bounded DP, with outcome range $R_{y_i}$, the per-area sensitivity is:
\begin{equation}\label{eq:hajek-sensitivity}
\Delta^{HJ}_{y_i} = \frac{w_{\max,i}\, R_{y_i}}{\sum_r w_{ir}}
\end{equation}
where $w_{\max,i} = \max_r w_{ir}$ is the largest survey weight in area $i$ and $\sum_r w_{ir}$ is the sum of all weights.
Under simple random sampling (SRS), every respondent has weight $w_{ir} = N_i/n_i$, so $\Delta_{y_i} = R_{y_i}/n_i$. For complex surveys with unequal weights, the sensitivity is amplified by the weight inequality ratio:
\begin{equation}\label{eq:hajek-wratio}
\Delta^{HJ}_{y_i} = \frac{w_{\max,i}/\bar{w}_i}{n_i}R_{y_i}
\end{equation}
where $\bar{w}_i = (\sum_r w_{ir})/n_i$ is the mean weight. The ratio $w_{\max,i}/\bar{w}_i$ captures how much the most influential respondent deviates from the average---a quantity determined entirely by the survey design.

Three quantities therefore govern the privacy guarantee: the sample size $n_i$, the weight inequality $w_{\max}/\bar{w}$, and the posterior variance $\sigma_i^2$ (which depends on the variance components $\sigma_{y_i}^2$ and $\sigma_v^2$). Equal-probability designs ($w_{\max}/\bar{w} = 1$) achieve the strongest privacy for a given sample size, while designs with highly variable weights can substantially inflate the privacy cost. Since the sensitivity enters $\rho_i$ squared, a design whose most influential respondent carries seven times the average weight, as in the ACS application below, has a $\rho_i$ roughly fifty times larger ($7^2 \approx 50$) than an equal-probability design with the same sample size.

For both estimators, the sensitivity scales linearly with $R_{y_i}$, and thus $\rho \propto R_y^2$.

\section{Applications}\label{sec:applications}
We illustrate the framework using two applications spanning different domains and sampling regimes. The first uses American Community Survey (ACS) microdata to estimate poverty prevalence across 2,462 PUMAs. The second uses Behavioral Risk Factor Surveillance System (BRFSS) data to estimate smoking prevalence across 52 substrata in Washington state. Both applications fit an intercept-only Fay-Herriot model with area-specific sampling variances, invoking Proposition~\ref{prop:intercept-varying}. In both applications the maximum in (\ref{eq:rho-per-area}) is attained at $k = i$ for every area, so the reported per-area guarantee is $\rho_i = c_{ii}^2\Delta_{y_i}^2/(2\sigma_i^2)$. Throughout, we refer to $\rho_i$ as the privacy loss of area $i$; smaller values mean stronger protection. We use the H\'{a}jek estimator in both applications as it is the standard choice for estimating means and proportions. For it does not require the population sizes $N_i$ to be known and it tends to have smaller sampling variance than the Horvitz-Thompson estimator even in non-private contexts \citep{lohr2021sampling}.
We treat the weights as fixed which means we implicitly pretend that these weights are design weights and do not include any nonresponse adjustments or calibration. Furthermore, we treat the REML estimates of $\sigma_v^2$ and $\sigma_{y_i}^2$ as invariants (Section~\ref{sec:robustness}). How to account for these weighting adjustment steps is currently an unsolved problem for DP \citep{drechsler2023differential}. The reported $\varepsilon$ values use $\delta < 1/N$ with $N$ the number of respondents, giving $\delta = 10^{-7}$ for the ACS and $\delta = 10^{-5}$ for the BRFSS; the $\rho$ values are $\delta$-free.

\subsection{Data and Methods}\label{sec:data}

\subsubsection{ACS Poverty Prevalence}\label{sec:acs-pums}

We analyze 2022 ACS 1-year Public Use Microdata Sample (PUMS) data from IPUMS USA \citep{ipums2025usa}, comprising 3.19 million person records across 2,462 Public Use Microdata Areas (PUMAs).

For each PUMA, we compute:
\begin{enumerate}
    \item \textbf{Direct estimate:} Weighted poverty rate using the H\'{a}jek estimator $\hat{y}_i = \sum_r w_{ir} Y_{ir} / \sum_r w_{ir}$, where $w_{ir}$ are the ACS person weights (PWGTP).
    \item \textbf{Sampling variance:} Design-based sampling variances $\sigma_{y_i}^2$ via successive difference replication using 80 replicate weights (PWGTP1--PWGTP80).
    \item \textbf{Sensitivity:} $\Delta_{y_i} = w_{\max,i} / \sum_r w_{ir}$, computed from the person weights. This requires that the per-area weight bounds $w_{\max,i}$ and weight totals $\sum_r w_{ir}$ are either published as survey design metadata or treated as public.\end{enumerate}

We fit an intercept-only Fay-Herriot model using REML via the \texttt{eblupFH} function in the \texttt{sae} package \citep{molina-marhuenda:2015}:
\begin{equation}
y_i = \mu + v_i + e_i, \quad v_i \sim N(0, \sigma^2_v), \quad e_i \sim N(0, \sigma^2_{y_i})
\end{equation}
where $y_i$ is the direct poverty rate estimate and $\sigma^2_{y_i}$ is the design-based sampling variance for PUMA $i$.

\subsubsection{BRFSS Smoking Prevalence}\label{sec:brfss}

To demonstrate the framework in a different domain and sampling regime, we apply it to
the 2022 Behavioral Risk Factor Surveillance System (BRFSS), estimating small-area
current smoking prevalence in Washington state. The BRFSS is the largest continuously
conducted health survey in the United States, collecting data on behavioral risk factors
through telephone interviews \citep{cdc2022brfss}. Small area estimation strategies based on BRFSS
data have been used in the public health literature to estimate county- and
community-level prevalence of smoking and other risk factors
\citep{dwyer2014cigarette, li2009small}.

We analyze 24,197 Washington state respondents across 52 geographic substrata
defined by the BRFSS sampling design (typically corresponding to counties or
groups of counties).\footnote{The public-use BRFSS file does not include county
FIPS codes. The substrata are identified via the \texttt{\_STSTR} variable
after stripping the phone-type digit.} The binary outcome is current smoking
status, defined as having smoked at least 100 cigarettes in one's lifetime and
currently smoking on some or all days (\texttt{\_SMOKER3} $\in \{1, 2\}$).
The overall smoking prevalence is 9.1\%.
As in the ACS analysis, we compute H\'{a}jek direct estimates and exact
sensitivities $\Delta_{y_i} = w_{\max,i}/\sum_r w_{ir}$ for each
substratum. 

\subsubsection{A Contrasting Regime}

The BRFSS application differs from the ACS in two important ways.
First, the per-area sample sizes are much smaller: median $n = 120$
versus 1,217 in the ACS.
Second, the weight inequality is higher: median $w_{\max}/\bar{w} = 8.24$
versus 6.92 in the ACS. Table~\ref{tab:setup} summarizes the
key differences.

\begin{table}[htbp]
\centering
\caption{Design and model characteristics of the two applications.}
\label{tab:setup}
\begin{tabular}{lcc}
\toprule
 & ACS (poverty) & BRFSS (smoking) \\
\midrule
Survey & ACS PUMS 2022 & BRFSS 2022 \\
Outcome & Poverty indicator & Current smoker \\
Geographic units & 2,462 PUMAs & 52 substrata \\
Respondents $N$ & 3,188,316 & 24,197 \\
Sample size $n_i$: min / median / max & 368 / 1,217 / 3,911 & 41 / 120 / 5,226 \\
Weight inequality $w_{\max,i}/\bar{w}_i$: median / max & 6.92 / 41.1 & 8.24 / 17.3 \\
\midrule
$\hat{\sigma}^2_v$ (REML) & $2.73 \times 10^{-3}$ & $8.00 \times 10^{-4}$ \\
Mean sampling variance $\bar{\sigma}^2_y$ & $3.92 \times 10^{-4}$ & $2.35 \times 10^{-3}$ \\
Shrinkage $B_i$: min / median / max & 0.004 / 0.101 / 0.55 & 0.03 / 0.569 / 0.94 \\
Median tightening factor $c_{ii}^2$ & 0.81 & 0.19 \\
\bottomrule
\end{tabular}
\end{table}

\subsection{Per-Area Privacy Guarantees}\label{sec:per-area}

Table~\ref{tab:privacy} summarizes the distribution of the per-area guarantee $\rho_i$ across areas, with $\varepsilon_i$ at the application's $\delta$ as a secondary summary. For comparison, the last row reports the median of the per-area bound $\Delta_{y_i}^2/(2\sigma_i^2)$ obtained by replacing $c_{ii}$ by its upper bound of one (the per-area analogue of the $S_y$-based bound in Section~\ref{sec:analytical}; we call it the conservative bound below), which is the guarantee an agency would compute without the exact coefficients of Proposition~\ref{prop:intercept-varying}.

\begin{table}[htbp]
\centering
\caption{Distribution of the per-area guarantee across areas. $\varepsilon_i$ is obtained from $\rho_i$ via (\ref{eq:zcdp-conversion}) at $\delta = 10^{-7}$ (ACS) and $\delta = 10^{-5}$ (BRFSS). The last row gives the median of the conservative bound $\Delta_{y_i}^2/(2\sigma_i^2)$, which ignores the tightening factor $c_{ii}^2$.}
\label{tab:privacy}
\begin{tabular}{lcccc}
\toprule
 & \multicolumn{2}{c}{ACS (poverty)} & \multicolumn{2}{c}{BRFSS (smoking)} \\
\cmidrule(lr){2-3}\cmidrule(lr){4-5}
 & $\rho_i$ & $\varepsilon_i$ & $\rho_i$ & $\varepsilon_i$ \\
\midrule
Minimum & 0.003 & 0.43 & 0.040 & 1.40 \\
10th percentile & 0.017 & 1.05 & 0.050 & 1.56 \\
Median & 0.048 & 1.81 & 0.188 & 3.13 \\
90th percentile & 0.163 & 3.40 & 2.49 & 13.2 \\
Maximum & 1.51 & 11.4 & 17.4 & 45.6 \\
\midrule
Areas with $\rho_i < 0.1$ & \multicolumn{2}{c}{79\%} & \multicolumn{2}{c}{27\%} \\
Areas with $\rho_i < 1$ & \multicolumn{2}{c}{99.7\%} & \multicolumn{2}{c}{81\%} \\
\midrule
Median conservative bound $\Delta_{y_i}^2/(2\sigma_i^2)$ & 0.061 & 2.05 & 3.08 & 15.0 \\
\bottomrule
\end{tabular}
\end{table}

The median $\rho_i = 0.048$ ($\varepsilon_i = 1.81$) in the ACS reflects the combined effect of the ACS's moderate sample sizes and substantial weight inequality. The BRFSS guarantees are weaker, with median $\rho_i = 0.188$ ($\varepsilon_i = 3.13$) and a long right tail driven by substrata with small samples and extreme weight inequality: the maximum, a substratum with 42 respondents and $w_{\max}/\bar{w} = 10.4$, is $\rho_i = 17.4$.

The conservative bound overstates the privacy loss by a factor of $1/c_{ii}^2$ (Section~\ref{sec:intercept-exact}) and depends on the shrinkage. This ratio has median 1.24 across PUMAs in the lightly shrunk ACS and median 5.1 across substrata in the heavily shrunk BRFSS. This is because $c_{ii}^2 \approx (1-B_i)^2$ is small for heavily shrunk areas. The effect on the distribution is larger still: the conservative bound is largest exactly in the small, heavily shrunk substrata where $c_{ii}^2$ is smallest, so its median in the BRFSS (3.08) is 16 times the median exact guarantee (0.188). Consequently, the difference between the two applications is smaller than their sample sizes alone would suggest. The typical BRFSS substratum has a tenth of the sample of the typical PUMA, but the ratio of the median guarantees is only about four.

\subsubsection{Role of Weight Inequality}

The dominant factor explaining the variation in $\rho_i$ across areas is the weight inequality ratio $w_{\max}/\bar{w}$ rather than the sample size. Combining (\ref{eq:hajek-wratio}) with (\ref{eq:rho-per-area}) gives $\log\rho_i = 2\log(w_{\max,i}/\bar{w}_i) - 2\log n_i + \log[c_{ii}^2/(2\sigma_i^2)]$. With other things held equal, doubling the weight inequality ratio quadruples the privacy loss. Figure~\ref{fig:rho-vs-n} plots $\rho_i$ against the sample size $n_i$ (both on the log scale), with each point colored by the area's weight inequality ratio $w_{\max,i}/\bar{w}_i$. In the ACS, sample sizes lie in the range 368--3,911 and on their own show no relationship with $\rho_i$ (a regression of $\log\rho_i$ on $\log n_i$ has slope 0.12, with 95\% CI: ($-0.003$, $0.23$)), partly because larger PUMAs also tend to have more unequal weights (correlation 0.44 between $\log n_i$ and $\log(w_{\max,i}/\bar{w}_i)$). Adding $\log(w_{\max,i}/\bar{w}_i)$ to the regression gives an estimated coefficient of 1.88 (95\% CI: ($1.76$, $1.99$)) on the log weight ratio and $-0.67$ (95\% CI: ($-0.78$, $-0.56$)) on $\log n_i$. The former is closer to the value of $2$ in the identity above. To quantify the cost of unequal weights directly, we recompute every guarantee with the weights within each area set equal to $\bar{w}_i$, holding $n_i$, $\sigma_i^2$, and $c_{ii}$ fixed. This divides $\rho_i$ by $(w_{\max,i}/\bar{w}_i)^2$. For ACS, the median of this factor across PUMAs is 48, and reduces the median $\rho_i$ from 0.048 to 0.0009. 

In the BRFSS, sample sizes span the much wider range 41--5,226 and do matter (coefficient $-0.41$, with 95\% CI ($-0.72$, $-0.09$), on $\log n_i$). The estimated coefficient of the log weight ratio, 1.09 (95\% CI: ($0.005$, $2.17$)), is less precisely determined with just 52 areas in this case. The largest privacy losses belong to substrata combining small samples with unequal weights, and equal weighting would divide $\rho_i$ by a median factor of 68.

\begin{figure}[htbp]
\centering
\includegraphics[width=\textwidth]{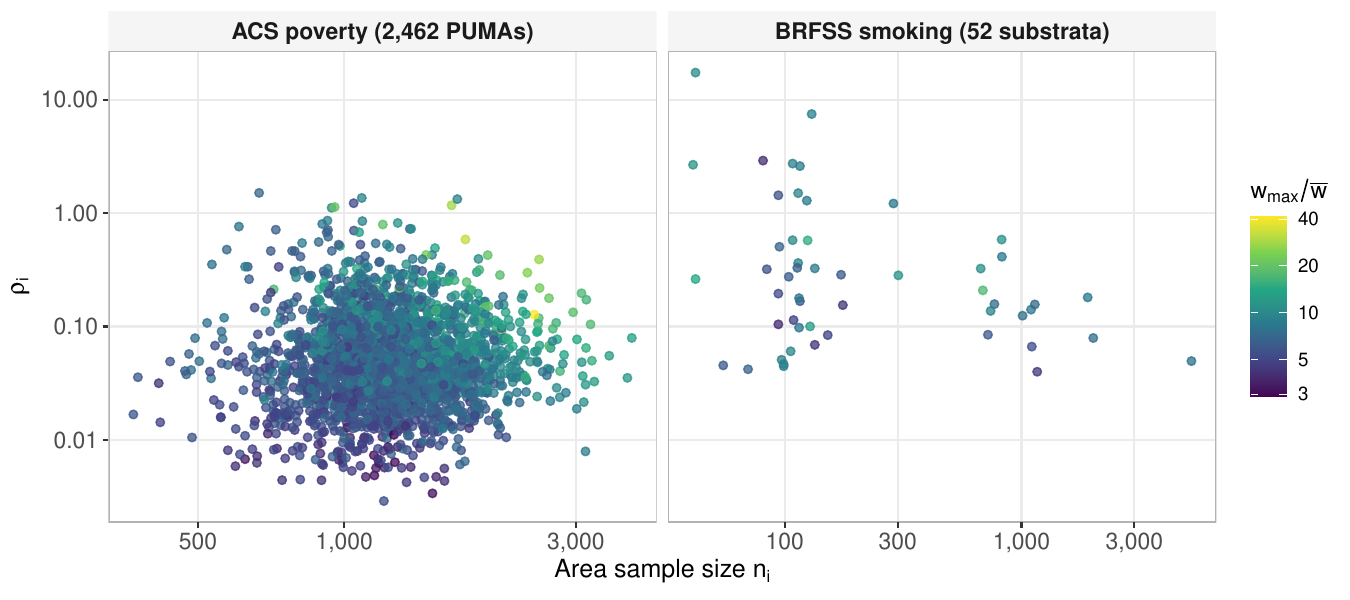}
\caption{Per-area guarantee $\rho_i$ versus area sample size $n_i$ (both on the log scale), one point per area, colored by the area's weight inequality ratio $w_{\max,i}/\bar{w}_i$. At any given sample size, the areas with the weakest guarantees (largest $\rho_i$) are those with the most unequal weights; in the BRFSS the largest values occur where small samples and unequal weights coincide.}
\label{fig:rho-vs-n}
\end{figure}

This finding has implications for survey design and output review:
\begin{enumerate}
    \item Equal-probability designs provide the strongest privacy for a given sample size.
    \item Weight trimming (capping extreme weights) directly reduces the sensitivity $\Delta_{y_i}$ and thus improves privacy, at the cost of introducing some bias in the direct estimates.
    \item When assessing the privacy of SAE outputs, agencies should report $w_{\max}/\bar{w}$ alongside the sample size, as the weight structure determines the actual privacy guarantee.
\end{enumerate}

\subsection{Composition Analysis}\label{sec:joint}

When releasing estimates for all areas simultaneously, the joint-release bound from Equation~\ref{eq:rho-joint} gives $\rho_{\mathrm{joint}} = 1.51$ for the 2,462 PUMAs ($\varepsilon_{\mathrm{joint}} = 11.4$ at $\delta = 10^{-7}$) for the ACS and $\rho_{\mathrm{joint}} = 17.8$ for the 52 substrata ($\varepsilon_{\mathrm{joint}} = 46.4$ at $\delta = 10^{-5}$) for the BRFSS. In both cases the maximum in (\ref{eq:rho-joint}) is attained at the area with the largest per-area guarantee, and the contribution of the other areas $r_k$ is negligible (Eq.~\ref{eq:rho-joint-decomp}): $r_k = 8 \times 10^{-7}$ at the maximizing PUMA ($\max_k r_k = 3.4 \times 10^{-5}$) and $r_k = 0.024$ at the maximizing substratum ($\max_k r_k = 0.042$). The joint bound is therefore driven by the single most sensitive area (the $k = i$ term) and exceeds the largest per-area guarantee by 2.4\% in the BRFSS and by a negligible amount in the ACS. Sequential composition (Section~\ref{sec:DP_properties}) of the same per-area guarantees would instead give $\sum_i \rho_i = 195$ ($\varepsilon = 307$) for the ACS and $49.6$ ($\varepsilon = 97$) for the BRFSS, overstating the joint loss by factors of 129 and 2.8; the gain from the joint bound grows with the number of areas.

\subsection{Robustness Validation}\label{sec:app-robustness}

The formal guarantees assume that the true variance components are known. In practice the variance components typically need to be estimated (see Section~\ref{sec:robustness} for a discussion of the privacy implications). Here we assess the numerical stability of the computed guarantees under potential measurement error of the estimated $\hat{\sigma}_v^2$. This is not a formal privacy analysis of the estimation step but rather a check how measurement errors will change the reported bound. To model the error, we replace $\hat{\sigma}_v^2$ with $\hat\sigma_v^2(1 + \eta)$ with $\eta \in \{\pm 0.1, \pm 0.25, \pm 0.5\}$ and recompute every per-area guarantee, which changes $B_i$, $W_+$, $c_{ii}$, and $\sigma_i^2$ simultaneously. Figure~\ref{fig:robustness} shows the distribution across areas of the relative change in $\rho_i$.

\begin{figure}[htbp]
\centering
\includegraphics[width=\textwidth]{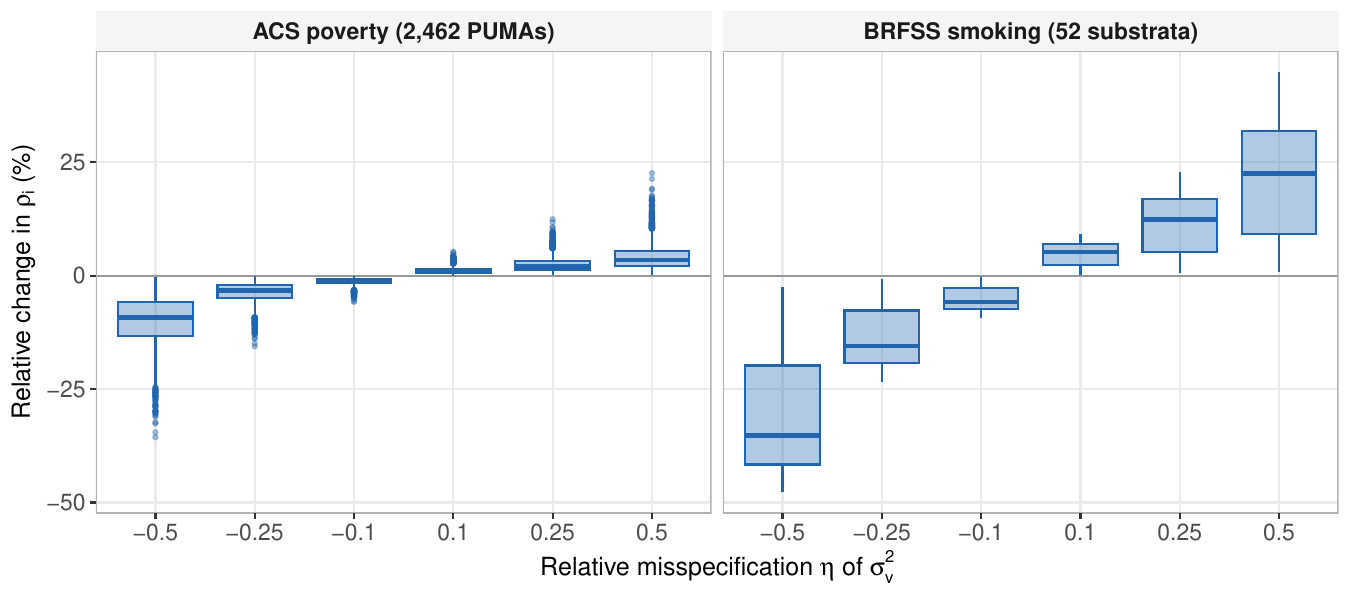}
\caption{Relative change in the per-area guarantee $\rho_i$ when the model variance is misspecified as $\hat{\sigma}_v^2(1+\eta)$. Boxes show the interquartile range across areas, whiskers extend to the most extreme values within 1.5 interquartile ranges, and points are areas beyond the whiskers.}
\label{fig:robustness}
\end{figure}

A larger model variance implies less shrinkage and a larger posterior variance; the first effect dominates, so $\rho_i$ increases with $\eta$ (weaker privacy) in every area. The magnitude follows the shrinkage. In the ACS, a 10\% error in $\hat{\sigma}_v^2$ changes the median $\rho_i$ by about 1\% and no PUMA's $\rho_i$ by more than 6\%. Even a 50\% error moves the median by less than 10\%, with the most affected PUMAs (those with the heaviest shrinkage) changing by up to $-36$\% and $+23$\%. In the BRFSS the guarantees are considerably more sensitive: a 10\% error changes $\rho_i$ by up to 9\%, and a 50\% error changes the median by $-35$\% or $+23$\% and individual substrata by up to $-48$\% and $+45$\%. 

\section{Discussion and Conclusions}\label{sec:discussion}

We have shown that the Bayesian Fay-Herriot model, when used to generate posterior draws for the small area posterior means, satisfies $\rho$-zCDP with $\rho_i \leq S_{y}^2/(2\sigma_i^2)$ (Equation~\ref{eq:czdp_guarantee}). This follows from the connection between the posterior sampling mechanism and the Gaussian mechanism (Section~\ref{sec:rdp}). For the intercept-only model, Proposition~\ref{prop:intercept-varying} provides exact coefficients $c_{ik}$ for how much the posterior mean can shift under area-specific variances for any finite $m$. The joint-release bound replaces sequential composition with a maximum over the  areas, reducing the composed guarantee from $\rho = 195$ (ACS) and $49.6$ (BRFSS) under sequential composition to $\rho = 1.5$ and $17.8$ under the Proposition's joint bound in our applications.

The most consequential finding for practice is the dominant role of weight inequality for the H\'{a}jek estimator. Under equal-probability sampling, the ACS data would yield a median $\rho_i \approx 0.0009$. The actual median of 0.048 is driven almost entirely by the weight inequality ratio $w_{\max}/\bar{w} \approx 7$. The BRFSS, with median $w_{\max}/\bar{w} \approx 8.2$, shows an even stronger effect. This suggests that agencies seeking stronger privacy guarantees for SAE outputs should consider weight trimming strategies, acknowledging the bias--privacy tradeoff this entails. A second finding, specific to the intercept-only model of Proposition~\ref{prop:intercept-varying}, is that shrinkage tightens the guarantee: the exact coefficient $c_{ii}^2 \approx (1-B_i)^2$ reduces the conservative bound by a median factor of 1.2 in the lightly shrunk ACS and 5.1 in the heavily shrunk BRFSS, at the price that the guarantee for heavily shrunk areas depends strongly on the estimated model variance.

Several limitations should be noted. First, the formal guarantees require that the variance components $\sigma_v^2$ and $\sigma_{y_i}^2$ are known. In practice, both are estimated from the data, which introduces privacy leakage not captured by our analysis. We adopt the invariant approach (Section~\ref{sec:robustness}), which restricts the domain of neighboring datasets to those yielding the same variance estimates. The most important direction for a future work is to quantify the leakage of the variance estimation step, or replacing the estimates by differentially private ones and composing the two costs (Options 1 and 2 of Section~\ref{sec:robustness}). This would turn the guarantee into an unconditional one. Second, the exact guarantees of Proposition~\ref{prop:intercept-varying} apply to the intercept-only model; extending them to a general design matrix with varying variances, as well as to proper priors, unit-level models, or non-Gaussian outcomes, requires separate analysis. Third, both the ACS PUMS and BRFSS data have themselves been subject to disclosure avoidance procedures (e.g., data swapping in the ACS, suppression of county identifiers in the BRFSS), so our computed privacy parameters describe the mechanism as applied to the public files.

The Fay-Herriot model is a standard two-level hierarchical model. The same structure arises in disease mapping, educational testing, and environmental monitoring. The analytical results in Section~\ref{sec:analytical} characterize the privacy guarantee as a function of the survey design parameters $(n_i, w_{\max}/\bar{w})$ and can be used by practitioners in any of these domains. This provides statistical agencies with a concrete, formula-based tool for assessing the inherent privacy properties of SAE outputs released as posterior draws, or as posterior means with Gaussian noise of variance $\sigma_i^2$ added (Section~\ref{sec:rdp}), and for identifying the design levers---particularly weight trimming and sample size---that could strengthen these guarantees.

\clearpage
\appendix
\setcounter{section}{0}\setcounter{equation}{0}\setcounter{table}{0}\setcounter{figure}{0}\setcounter{lemma}{0}
\renewcommand{\thesection}{S\arabic{section}}
\renewcommand{\thetable}{S\arabic{table}}
\renewcommand{\thefigure}{S\arabic{figure}}
\renewcommand{\theequation}{S\arabic{equation}}
\renewcommand{\thelemma}{S\arabic{lemma}}
\begin{center}\Large\bfseries Supplement to ``Differential Privacy Guarantees in Small Area Estimation''\end{center}
\medskip
This supplement, appended here to the main paper, contains the proofs of the results in the main paper (Section~\ref{sec:proofs}) and additional tables and figures for the two applications (Section~\ref{sec:extra}). Equation, theorem, and section numbers without the prefix S refer to the main paper.

\section{Proofs}\label{sec:proofs}

Throughout, $D$ and $D'$ are neighboring databases under bounded differential privacy that differ in one record in area $k$, so that $y_i^{(D)} = y_i^{(D')}$ for $i \neq k$, the auxiliary variables $x_i$ are public, and the variance components are known. We write $\mu_i^D$ and $\sigma_i^2$ for the posterior mean and variance of $\theta_i$ under $D$ (Equation~(16) of the main paper and the display that follows it); $\sigma_i^2$ depends only on $x_i$ and the variance components and is therefore the same under $D$ and $D'$.

\subsection{Proof of Theorem 1}\label{sec:proof-thm1}

\begin{theoremA}
For any two neighboring databases $D$ and $D'$ that differ in one record in area $k$, the privacy loss random variable of the mechanism $\Theta_i$ under $\tilde{\theta}_i \sim N(\mu_i^D, \sigma_i^2)$ satisfies
\begin{equation*}
\mathcal{L}(\tilde{\theta}_i) \sim N\left(\frac{(\mu_i^D - \mu_i^{D'})^2}{2\sigma_i^2}, \frac{(\mu_i^D - \mu_i^{D'})^2}{\sigma_i^2}\right).
\end{equation*}
\end{theoremA}

Theorem 1 is a distributional statement about the privacy loss random variable for two equal-variance Gaussian distributions.

\begin{lemma}[Privacy loss for equal-variance Gaussian distributions]\label{lem:privacy-loss-gaussian}
Let $P_D = N(\mu^D, \sigma^2)$ and $P_{D'} = N(\mu^{D'}, \sigma^2)$. For a draw $\tilde{\theta} \sim P_D$, the privacy loss random variable $\mathcal{L}(\tilde{\theta}) = \log \{f_D(\tilde{\theta})/f_{D'}(\tilde{\theta})\}$ follows a Gaussian distribution:
\begin{equation}
\mathcal{L}(\tilde{\theta}) \sim N\left(\frac{(\mu^D - \mu^{D'})^2}{2\sigma^2}, \frac{(\mu^D - \mu^{D'})^2}{\sigma^2}\right).
\end{equation}
\end{lemma}

\begin{proof}
The log-likelihood ratio is
\begin{align}
\mathcal{L}(\tilde{\theta}) &= \log \frac{\exp\left(-\frac{(\tilde{\theta} - \mu^D)^2}{2\sigma^2}\right)}{\exp\left(-\frac{(\tilde{\theta} - \mu^{D'})^2}{2\sigma^2}\right)}
= \frac{(\tilde{\theta} - \mu^{D'})^2 - (\tilde{\theta} - \mu^D)^2}{2\sigma^2}
= \frac{\tilde{\theta}(\mu^D - \mu^{D'})}{\sigma^2} - \frac{(\mu^D)^2 - (\mu^{D'})^2}{2\sigma^2}. \label{eq:privacy-loss-linear}
\end{align}
Since $\tilde{\theta} \sim N(\mu^D, \sigma^2)$ and \eqref{eq:privacy-loss-linear} is linear in $\tilde{\theta}$, $\mathcal{L}$ is Gaussian with
\begin{align}
\mathbb{E}[\mathcal{L}] &= \frac{\mu^D(\mu^D - \mu^{D'})}{\sigma^2} - \frac{(\mu^D)^2 - (\mu^{D'})^2}{2\sigma^2} = \frac{(\mu^D - \mu^{D'})^2}{2\sigma^2}, \\
\mathrm{Var}[\mathcal{L}] &= \frac{(\mu^D - \mu^{D'})^2}{\sigma^4} \cdot \sigma^2 = \frac{(\mu^D - \mu^{D'})^2}{\sigma^2}.
\end{align}
\end{proof}

Applying Lemma~\ref{lem:privacy-loss-gaussian} with $\mu^D = \mu_i^D$, $\mu^{D'} = \mu_i^{D'}$, and $\sigma^2 = \sigma_i^2$ gives the distribution stated in Theorem 1. The expected privacy loss is the mean of this distribution, namely the Kullback--Leibler divergence
\begin{equation}\label{eq:kl-gaussian}
D_{\mathrm{KL}}(P_D \| P_{D'}) = \mathbb{E}_{\tilde{\theta} \sim P_D}[\mathcal{L}(\tilde{\theta})] = \frac{(\mu_i^D - \mu_i^{D'})^2}{2\sigma_i^2}.
\end{equation}
Because $\mathcal{L}(\tilde{\theta}_i)$ is Gaussian, it has unbounded support, so no finite $\varepsilon$ can bound it almost surely; this is the impossibility of pure $\varepsilon$-DP discussed in Section 3.2 of the main paper.

\subsection{Proof of Theorem 2}\label{sec:proof-thm2}

\begin{theoremB}
Under the equal-variance assumption $\sigma_{y_i}^2 = \sigma_y^2$, for any two neighboring databases $D$ and $D'$ that differ in one record in area $k$, the mechanism $\Theta_i$ satisfies $(\alpha, \varepsilon_i^{(\alpha)})$-RDP for any $\alpha > 1$, with
\begin{equation*}
\varepsilon_i^{(\alpha)}\leq\frac{\alpha  [\max\{\Delta_{y_i};B\,S_y\} ]^2}{2\sigma_i^2}\leq\frac{\alpha S^2_{y}}{2\sigma_{i}^2},
\end{equation*}
where $\Delta_{y_i}=\max_{D \sim D'} |y_i^{(D)} - y_i^{(D')}|$ and $S_{y} =  \max_i\Delta_{y_i}$. Under bounded DP with fixed weights, $\Delta_{y_i}^{HT}= w_{\max,i}R_{y_i}/N_i$ and $\Delta_{y_i}^{HJ}= w_{\max,i}R_{y_i}/\sum_r w_{ir}$.
\end{theoremB}

We first record the R\'enyi divergence between two equal-variance Gaussians (Lemma~\ref{lem:renyi_normal}), then bound the sensitivity of the posterior mean, and combine the two.

\subsubsection{R\'enyi divergence for Gaussian distributions}

\begin{lemma}[R\'enyi divergence for equal-variance Gaussian distributions]\label{lem:renyi_normal}
For $\alpha > 1$,
\begin{equation}
    D_{\alpha}\big(N(\mu_1, \sigma^2)\,\|\,N(\mu_2, \sigma^2)\big) = \frac{\alpha(\mu_1 - \mu_2)^2}{2\sigma^2}.
\end{equation}
\end{lemma}

\begin{proof}
Using the definition of R\'enyi divergence:
\begin{align*}
    D_{\alpha}(N(\mu_1, \sigma^2)\|N(\mu_2, \sigma^2)) &= \frac{1}{\alpha - 1} \log \int_{-\infty}^{\infty}\frac{1}{\sigma \sqrt{2\pi}} \exp \left(-\frac{\alpha}{2\sigma^2}(x - \mu_1)^2 - \frac{1 - \alpha}{2\sigma^2}(x - \mu_2)^2\right) dx \\
    &= \frac{1}{\alpha - 1} \log \frac{1}{\sigma\sqrt{2\pi}} \int_{-\infty}^{\infty} \exp \left( -\frac{1}{2\sigma^2}(x - \alpha \mu_1 - (1 - \alpha)\mu_2)^2\right) \times \\
    & \quad \quad  \quad \quad \quad \exp \left( \frac{1}{2\sigma^2}(\alpha\mu_1 + (1 - \alpha)\mu_2)^2 - \alpha\mu_1^2 - (1 - \alpha)\mu_2^2 \right) dx \\
    &= \frac{1}{\alpha - 1}\log \exp \left(\frac{(\mu_1^2 + \mu_2^2) \alpha (\alpha - 1) - 2\mu_1\mu_2\alpha (\alpha - 1)}{2 \sigma^2} \right) \\
    &= \frac{\alpha(\mu_1 - \mu_2)^2}{2\sigma^2}.
\end{align*}
\end{proof}

\subsubsection{Sensitivity of the posterior mean}\label{sec:Sensitivity_Posterior_Mean}

To apply Lemma~\ref{lem:renyi_normal} we bound the shift of the posterior mean $|\mu_i^D - \mu_i^{D'}|$. Under the equal-variance assumption,
\begin{equation}
\mu_i^D = \frac{\sigma_v^2 y_i^{(D)} + \sigma_y^2 x_i^T\hat{\beta}_{WLS}^{(D)}}{\sigma_v^2 + \sigma_y^2},
\quad
\mu_i^D - \mu_i^{D'} = \frac{\sigma_v^2 (y_i^{(D)} - y_i^{(D')}) + \sigma_y^2 \,(x_i^T\hat{\beta}_{WLS}^{(D)} - x_i^T\hat{\beta}_{WLS}^{(D')})}{\sigma_v^2 + \sigma_y^2}.
\end{equation}
We bound the two contributions separately and use the triangle inequality to obtain a conservative upper bound for the sum.

\paragraph{Direct component ($y_i^{(D)} - y_i^{(D')}$).}
The sensitivity of the direct estimate is $\Delta_{y_i}=\max_{D \sim D'}|y_i^{(D)} - y_i^{(D')}|$. The difference is nonzero only when $i=k$: substituting one record $y_{ir} \to y'_{ir}$ shifts the weighted total by $w_{ir}(y_{ir} - y'_{ir})$, which is largest when the record carrying the maximum weight moves across the full range of possible outcome values. For the Horvitz--Thompson estimator,
\[
\Delta_{y_i}^{HT} =\max_{D\sim D'}\left|\frac{\sum_r w_{ir}y_{ir}}{N_i}-\frac{\sum_r w_{ir}y'_{ir}}{N_i}\right|=\frac{\max_r w_{ir}\cdot R_{y_i}}{N_i}
=\frac{w_{i}^{max}R_{y_i}}{N_i},
\]
where $w_{i}^{max}=\max_r w_{ir}$ is the maximum weight in area $i$ and $R_{y_i}$ is the range of possible values of $y$ in area $i$ ($R_{y_i}=1$ for a binary outcome). Note that $R_{y_i}$ is the range of the outcome's domain, not the observed sample range: a neighboring dataset may replace a record with a value outside the observed range, and the observed range itself changes across $D \sim D'$. For the H\'ajek estimator,
\[
\Delta_{y_i}^{HJ} =\max_{D\sim D'}\left|\frac{\sum_r w_{ir}y_{ir}}{\sum_r w_{ir}}-\frac{\sum_r w_{ir}y'_{ir}}{\sum_r w_{ir}}\right|=\frac{w_{i}^{max}R_{y_i}}{\sum_r w_{ir}}.
\]

\paragraph{Synthetic component ($x_i^T\hat{\beta}_{WLS}^{(D)} - x_i^T\hat{\beta}_{WLS}^{(D')}$).}
\begin{lemma}[WLS sensitivity]\label{lem:wls-sensitivity}
With $x_i$ public and $D, D'$ differing in area $k$,
\begin{equation}
\left|x_i^T(\hat{\beta}_{WLS}^{(D)} - \hat{\beta}_{WLS}^{(D')})\right| \leq \Delta_{y_k}.
\end{equation}
\end{lemma}
\begin{proof}
We have
\begin{align}
x_i^T(\hat{\beta}_{WLS}^{(D)} - \hat{\beta}_{WLS}^{(D')}) &= x_i^T\left[\sum_{j=1}^m (1-B)x_j x_j^T\right]^{-1} (1-B)x_k (y_k^{(D)} - y_k^{(D')}) \nonumber \\
&\leq \left|x_i^T\left[\sum_{j=1}^m x_j x_j^T\right]^{-1} x_k\right| \cdot \Delta_{y_k}. \label{eq:wls-sensitivity}
\end{align}
Let $u_i = (X^TX)^{-1/2}x_i$ and $u_k = (X^TX)^{-1/2}x_k$, so $x_i^T(X^TX)^{-1} x_k = u_i^T u_k$. By Cauchy--Schwarz,
\begin{equation}
|u_i^T u_k| \leq \|u_i\| \cdot \|u_k\| = \sqrt{x_i^T(X^TX)^{-1}x_i} \cdot \sqrt{x_k^T(X^TX)^{-1}x_k}.
\end{equation}
The terms $h_{ii} = x_i^T(X^TX)^{-1}x_i$ are diagonal entries of the hat matrix $H = X(X^TX)^{-1}X^T$, which is idempotent and symmetric, so $0 \leq h_{ii} \leq 1$~\citep{hoaglin1978hat}. Hence $\left|x_i^T(X^TX)^{-1} x_k\right| \leq \sqrt{h_{ii}h_{kk}} \leq 1$, and $|x_i^T(\hat{\beta}_{WLS}^{(D)} - \hat{\beta}_{WLS}^{(D')})| \leq \Delta_{y_k}$.
\end{proof}

\paragraph{Sensitivity of $\mu_i$.}
We distinguish two scenarios. If the area of interest is the one that changes between $D$ and $D'$ ($i = k$), both components contribute:
\begin{equation}
\left|\mu_k^D - \mu_k^{D'}\right| \leq \frac{\sigma_v^2 \, \Delta_{y_k} + \sigma_y^2 \, \Delta_{y_k}}{\sigma_v^2 + \sigma_y^2} =(1-B)\Delta_{y_k}+B\Delta_{y_k}= \Delta_{y_k}.
\end{equation}
For all other means $\mu_i$ $(i \neq k)$, only the synthetic component contributes:
\begin{equation}
\left|\mu_i^D - \mu_i^{D'}\right| = \frac{\sigma_y^2 \,|x_i^T(\hat{\beta}_{WLS}^{(D)} - \hat{\beta}_{WLS}^{(D')})|}{\sigma_v^2 + \sigma_y^2} \leq \frac{\sigma_y^2 \, \Delta_{y_k}}{\sigma_v^2 + \sigma_y^2} = B\, \Delta_{y_k}.
\end{equation}
The sensitivity is the maximum over both scenarios, noting that in the second the area that changes can be any of the other areas. With $S_y=\max_i\Delta_{y_i}$ and bounding $\max_{k\neq i}\Delta_{y_k}$ by $S_y$,
\begin{equation}\label{eq:sens_post_mean}
\Delta_{\mu_i}\leq\max\{\Delta_{y_i};B\,S_y\}\leq S_y.
\end{equation}

\paragraph{Conclusion.}
Since $\Theta_i(D) \sim N(\mu_i^D, \sigma_i^2)$ and $\Theta_i(D') \sim N(\mu_i^{D'}, \sigma_i^2)$ with the same variance, Lemma~\ref{lem:renyi_normal} gives $D_\alpha(\Theta_i(D)\,\|\,\Theta_i(D')) = \alpha(\mu_i^D - \mu_i^{D'})^2/(2\sigma_i^2) \leq \alpha \Delta_{\mu_i}^2/(2\sigma_i^2)$, and \eqref{eq:sens_post_mean} yields the bound in Theorem 2. The zCDP statement in Equation~(22) of the main paper follows because the bound is linear in $\alpha$.

\subsection{Proof of Proposition 1}\label{sec:proof-prop1}

\begin{propositionA}
Consider the intercept-only Fay--Herriot model with area-specific sampling variances $\sigma_{y_i}^2$, shrinkage factors $B_i = \sigma_{y_i}^2/(\sigma_{y_i}^2 + \sigma_v^2)$, and $W_+ = \sum_{j=1}^m(1-B_j)$. For any neighboring pair differing in one record in area $k$, $\mu_i^D - \mu_i^{D'} = c_{ik}\,(y_k^{(D)} - y_k^{(D')})$ with $c_{kk} = (1-B_k)(1 + B_k/W_+)$ and $c_{ik} = B_i(1-B_k)/W_+$ for $i \neq k$. Consequently (i) $\Theta_i$ satisfies $\rho_i$-zCDP with $\rho_i = [\max_{k} (c_{ik}\Delta_{y_k})]^2/(2\sigma_i^2) \leq S^2_y/(2\sigma^2_i)$, and (ii) $\Theta = (\Theta_1, \ldots, \Theta_m)$ satisfies $\rho_{\mathrm{joint}}$-zCDP with $\rho_{\mathrm{joint}} =  \max_{k} (\Delta_{y_k}^2/2) \sum_{i=1}^m c_{ik}^2/\sigma_i^2$. The coefficients satisfy $c_{kk} \leq 1$ and $c_{ik} \leq B_i$ for $i \neq k$.
\end{propositionA}

We assume throughout that the model does not include any area-level covariates beyond the intercept ($x_i = 1$ for all $i$), so the WLS estimator reduces to the scalar $\hat{\beta}_0^{Int} = \sum_j (1-B_j)\,y_j / W_+$, and the posterior mean for area $i$ is $\mu_i^{Int} = (1-B_i)\,y_i + B_i\,\hat{\beta}_0^{Int}$. The superscript \textit{Int} emphasizes that the results are specific to the random intercept model.

Since the posterior variance $\sigma_i^2$ depends only on the variance components $\{\sigma_{y_j}^2, \sigma_v^2\}$, which are assumed known, it does not change over neighboring datasets. Hence Lemma~\ref{lem:renyi_normal} still applies and we only need the shift of the posterior mean $\mu_i^{Int}$.

\subsubsection{Sensitivity of $\mu_i^{Int}$}

Similar to the proof strategy in Section~\ref{sec:Sensitivity_Posterior_Mean}, we derive upper bounds for the direct and the synthetic component separately and obtain a conservative bound for the sensitivity by using the triangle inequality.

\paragraph{Direct component ($y_i^{(D)} - y_i^{(D')}$).}
The sensitivity of the direct estimate remains $\Delta_{y_i}=\max_{D \sim D'}|y_i^{(D)} - y_i^{(D')}|$.

\paragraph{Synthetic component ($\hat{\beta}_0^{Int(D)} - \hat{\beta}_0^{Int(D')}$).}
For the intercept-only model and area specific variances, changing one record in area $k$ shifts the WLS intercept by
\begin{equation}
\left|\hat{\beta}_0^{Int(D)} - \hat{\beta}_0^{Int(D')}\right| = \frac{(1-B_k)}{W_+}\,\left|y_k^{(D)} - y_k^{(D')}\right| \;\leq\; \frac{(1-B_k)}{W_+}\,\Delta_{y_k}.
\end{equation}

\paragraph{Sensitivity for $\mu_i^{Int}$.}
To compute the sensitivity for the posterior mean $\mu_i^{Int}$, we again need to distinguish two scenarios. Assuming the area of interest is the one that changes between $D$ and $D'$ (i.e. $\mu_i^{Int}=\mu_k^{Int}$), both components contribute:
\begin{equation}
\left|\mu_k^{Int(D)} - \mu_k^{Int(D')}\right| \leq \underbrace{(1-B_k)\,\Delta_{y_k}}_{\text{direct}} + \underbrace{B_k\,\frac{(1-B_k)}{W_+}\,\Delta_{y_k}}_{\text{synthetic}} = (1-B_k)\!\left(1 + \frac{B_k}{W_+}\right)\Delta_{y_k} = c_{kk}\,\Delta_{y_k}.
\end{equation}
For all other small area estimates $\mu_i^{Int}$ $(i \neq k)$, only the synthetic component contributes.
\begin{equation}
\left|\mu_i^{Int(D)} - \mu_i^{Int(D')}\right| \leq B_i\,\frac{(1-B_k)}{W_+}\,\Delta_{y_k} = c_{ik}\,\Delta_{y_k}.
\end{equation}
To obtain the overall sensitivity, we need to find the maximum over both scenarios, noting that for the second case the area that changes can be any of the other areas. Thus, the sensitivity is given as
\begin{equation}\label{eq:sens_intercept-rw}
\Delta_{\mu_i^{Int}}=\max_{1 \leq k \leq m} c_{ik}\,\Delta_{y_k}.
\end{equation}

\subsubsection{Relationship to the results for the general model with constant variance}
We can compare the results for the intercept-only model with area specific variances to the results for the model with covariates but fixed variance from Section~\ref{sec:Sensitivity_Posterior_Mean}.

\medskip
\noindent\textbf{Coefficient bounds.}
Since every summand of $W_+$ is positive, $W_+ \geq (1-B_k)$, hence $(1-B_k)/W_+ \leq 1$ and
\begin{equation}
c_{kk} = (1-B_k) + \frac{B_k(1-B_k)}{W_+} \leq (1-B_k) + B_k = 1, \qquad c_{ik} = \frac{B_i(1-B_k)}{W_+} \leq B_i.
\end{equation}
Let $S_y=\max_i\Delta_{y_i}$. A conservative upper bound for the sensitivity is given as
\[
\Delta_{\mu_i^{Int}}=\max\{\Delta_{y_i};\, B_i S_y\}\leq S_y,
\]
i.e., the conservative upper bound is comparable to the fixed variance results for $\Delta_{\mu_i}$ from Eq.~(\ref{eq:sens_post_mean}) with the only difference that the constant term $BS_y$ is replaced by $B_i S_y$ to account for the fact that the shrinkage factor $B_i$ varies between the areas.

\subsubsection{Privacy guarantees for the intercept-only model}
\medskip
\noindent\textbf{Per-area privacy guarantee (claim (i)).}
By Lemma~\ref{lem:renyi_normal}, Eq.~(\ref{eq:sens_intercept-rw}) implies that under $\rho$-zCDP the privacy guarantee for a random draw of $\Theta_i^{Int}$ is given by
\begin{equation}
\rho_i^{Int}= \frac{\{\max_{1 \leq k \leq m}(c_{ik}\Delta_{y_k})\}^2}{2\sigma_i^2}\leq\frac{ S_y^2}{2\sigma_i^2}.
\end{equation}

\medskip
\noindent\textbf{Joint release (claim (ii)).}
Using the composition properties of $\rho$-zCDP, a conservative upper bound for the privacy loss of the simultaneous release of all $m$ small area estimates would be $\rho_{joint}^{Int}=\sum_i \rho_i^{Int}.$ However, a tighter bound can be obtained by directly computing the maximum of the R\'enyi divergence for the joint release. The $m$ posterior draws are independent, so the R\'enyi divergence is additive over the product measure. For a fixed neighboring pair changing one record in area $k$, write $\Delta_k = y_k^{(D)} - y_k^{(D')}$ for the realized shift of the direct estimate. Then
\begin{equation}
D_\alpha\!\left(\Theta(D)\,\|\,\Theta(D')\right) \leq \sum_{i=1}^m \frac{\alpha\,c_{ik}^2\,\Delta_k^2}{2\sigma_i^2} = \alpha\,\Delta_k^2 \sum_{i=1}^m \frac{c_{ik}^2}{2\sigma_i^2}.
\end{equation}
Bounding $|\Delta_k| \leq \Delta_{y_k}$ and maximizing over area $k$ yields
\begin{equation}
\rho_{\mathrm{joint}}^{Int} = \max_{1 \leq k \leq m} \frac{\Delta_{y_k}^2}{2}\sum_{i=1}^m \frac{c_{ik}^2}{\sigma_i^2},
\end{equation}
which is claim (ii). The bound is a maximum over $k$ rather than a sum because, under bounded DP, a neighboring pair changes one record and therefore one area; sequential composition, which bounds the loss of $m$ separate mechanisms each of which may be affected by the change, is the special case obtained by replacing every $c_{ik}$ by its maximum over $k$ and summing. Splitting the sum into its $i=k$ term and the remainder gives the decomposition in Equation~(26) of the main paper.

\section{Additional Results for the Applications}\label{sec:extra}

\subsection{Exact versus conservative per-area guarantees}

Table~\ref{tab:cons} compares the distribution of the exact per-area guarantee $\rho_i = c_{ii}^2 \Delta_{y_i}^2/(2\sigma_i^2)$ of Proposition 1(i) with the conservative bound $\Delta_{y_i}^2/(2\sigma_i^2)$ obtained from $c_{kk} \leq 1$. In both applications the maximum in Proposition 1(i) is attained at $k = i$ for every area; the largest value of $\max_{k \neq i} c_{ik}\Delta_{y_k}/(c_{ii}\Delta_{y_i})$ over areas is $1.5 \times 10^{-3}$ in the ACS and $0.57$ in the BRFSS. Figure~\ref{fig:tightness} plots the ratio $c_{ii}^2$ against the shrinkage factor $B_i$; the ratio is determined by $B_i$ up to the small correction $1 + B_i/W_+$, and it approaches $(1-B_i)^2$ as the number of areas grows.

\begin{table}[htbp]
\centering
\caption{Distribution across areas of the exact per-area guarantee $\rho_i$ and of the conservative bound $\Delta_{y_i}^2/(2\sigma_i^2)$, with $\varepsilon$ at $\delta = 10^{-7}$ (ACS) and $\delta = 10^{-5}$ (BRFSS).}
\label{tab:cons}
\begin{tabular}{lcccccccc}
\toprule
 & \multicolumn{4}{c}{ACS (poverty)} & \multicolumn{4}{c}{BRFSS (smoking)} \\
\cmidrule(lr){2-5}\cmidrule(lr){6-9}
 & \multicolumn{2}{c}{exact} & \multicolumn{2}{c}{conservative} & \multicolumn{2}{c}{exact} & \multicolumn{2}{c}{conservative} \\
 & $\rho_i$ & $\varepsilon_i$ & $\rho_i$ & $\varepsilon_i$ & $\rho_i$ & $\varepsilon_i$ & $\rho_i$ & $\varepsilon_i$ \\
\midrule
Minimum & 0.003 & 0.43 & 0.005 & 0.56 & 0.040 & 1.40 & 0.052 & 1.60 \\
10th percentile & 0.017 & 1.05 & 0.025 & 1.29 & 0.050 & 1.56 & 0.188 & 3.13 \\
Median & 0.048 & 1.81 & 0.061 & 2.05 & 0.188 & 3.13 & 3.08 & 15.0 \\
90th percentile & 0.163 & 3.40 & 0.183 & 3.61 & 2.49 & 13.2 & 9.34 & 30.1 \\
Maximum & 1.51 & 11.4 & 1.55 & 11.5 & 17.4 & 45.6 & 72.1 & 129.7 \\
\midrule
Median of ratio conservative/exact ($1/c_{ii}^2$) & \multicolumn{4}{c}{1.24} & \multicolumn{4}{c}{5.13} \\
\bottomrule
\end{tabular}
\end{table}

\begin{figure}[htbp]
\centering
\includegraphics[width=0.7\textwidth]{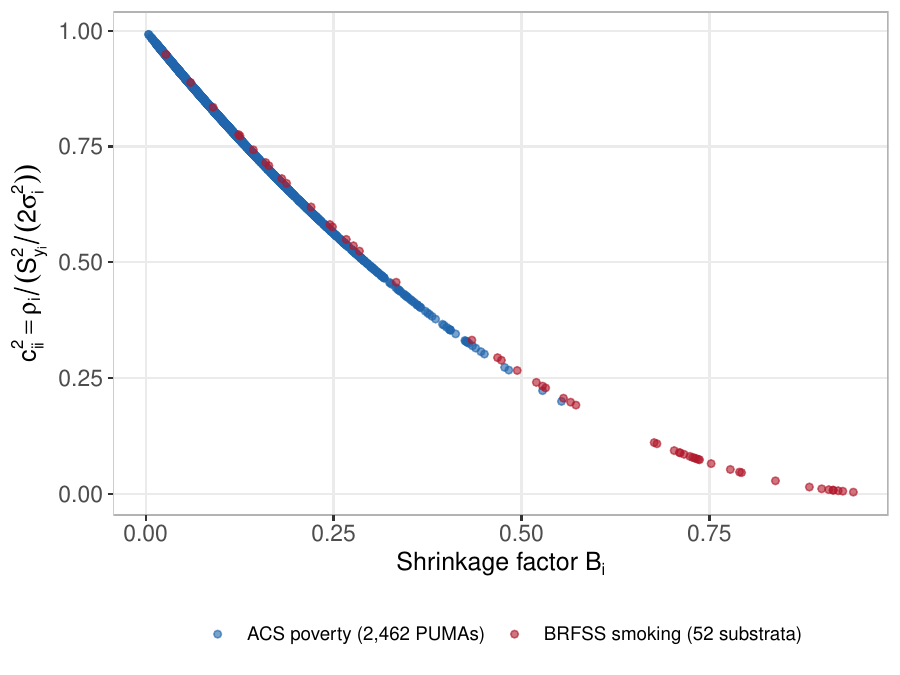}
\caption{Tightening factor $c_{ii}^2$ of the exact per-area guarantee relative to the conservative bound, against the shrinkage factor $B_i$, for all areas of both applications.}
\label{fig:tightness}
\end{figure}

\subsection{Weight inequality}

Figure~\ref{fig:wratio} plots $\rho_i$ against the weight inequality ratio $w_i^{max}/\bar{w}_i$. The correlation of $\log \rho_i$ with $\log(w_i^{max}/\bar{w}_i)$ is 0.50 in the ACS and 0.25 in the BRFSS, where sample size varies over a wider range.

\begin{figure}[htbp]
\centering
\includegraphics[width=\textwidth]{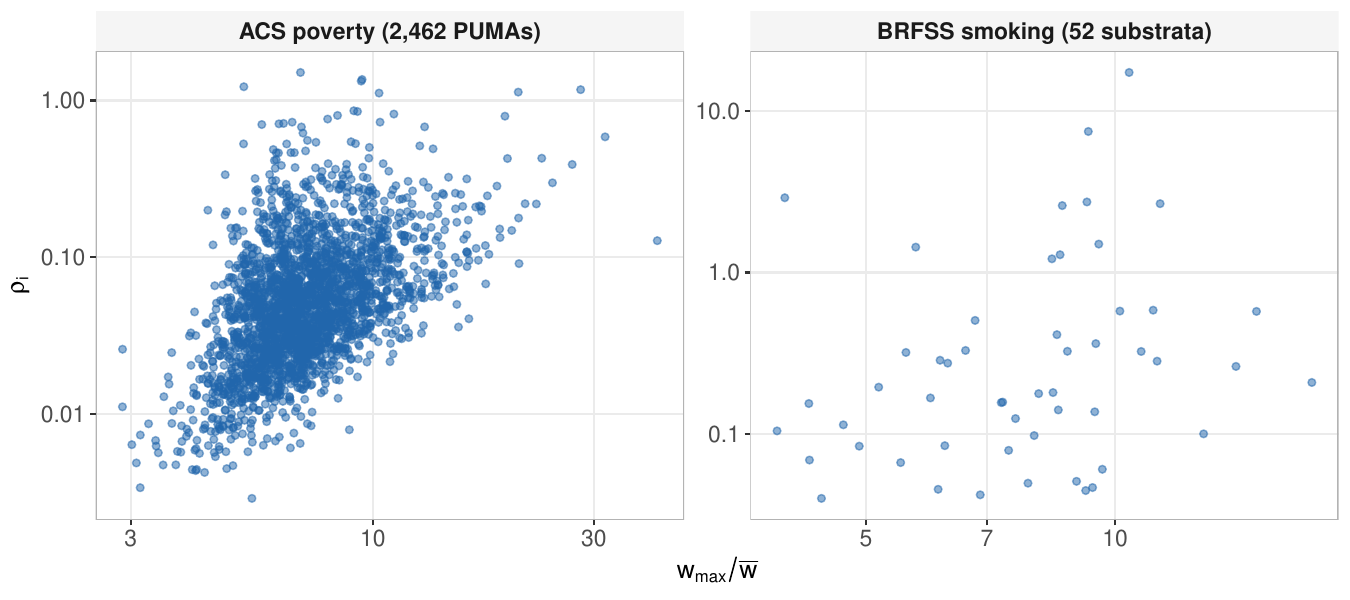}
\caption{Per-area guarantee $\rho_i$ versus the weight inequality ratio $w_i^{max}/\bar{w}_i$ (both on the log scale).}
\label{fig:wratio}
\end{figure}

\subsection{Dependence of $\varepsilon$ on $\delta$}

Figure~\ref{fig:epsdelta} shows the $(\varepsilon,\delta)$-DP guarantee implied by the $\delta$-free $\rho$ values via $\varepsilon = \rho + 2\sqrt{\rho\log(1/\delta)}$, for the per-area median and the joint release of each application. Since $\rho$ does not depend on $\delta$, the vertical gap between the two applications at any $\delta$ reflects only the difference in $\rho$. The points mark each application's $\delta < 1/N$.

\begin{figure}[htbp]
\centering
\includegraphics[width=0.75\textwidth]{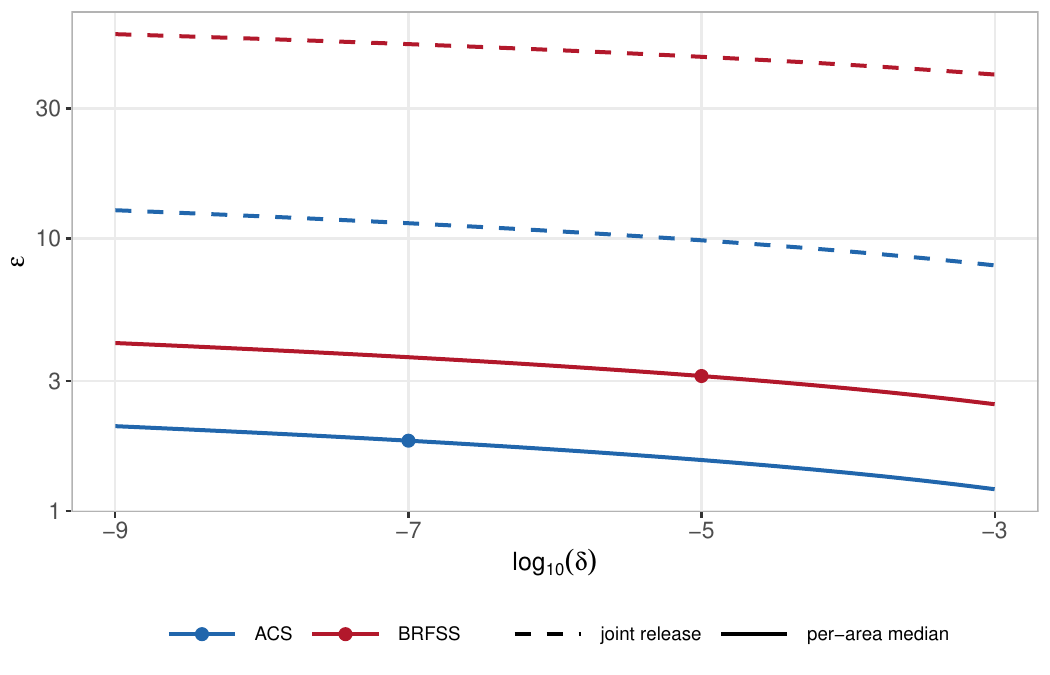}
\caption{$\varepsilon$ as a function of $\delta$ at the per-area median $\rho$ (solid) and at $\rho_{\mathrm{joint}}$ (dashed). Points mark $\delta = 10^{-7}$ (ACS) and $\delta = 10^{-5}$ (BRFSS).}
\label{fig:epsdelta}
\end{figure}

\subsection{Sensitivity to the estimated model variance}

Table~\ref{tab:rob} gives the quantiles across areas of the relative change in $\rho_i$ when $\sigma_v^2$ is replaced by $\hat{\sigma}_v^2(1+\eta)$, underlying Figure~2 of the main paper.

\begin{table}[htbp]
\centering
\caption{Relative change in the per-area guarantee $\rho_i$ (in \%) under misspecification $\hat{\sigma}_v^2(1+\eta)$ of the model variance: quantiles across areas.}
\label{tab:rob}
\begin{tabular}{lrrrrrrrrrr}
\toprule
 & \multicolumn{5}{c}{ACS (poverty)} & \multicolumn{5}{c}{BRFSS (smoking)} \\
\cmidrule(lr){2-6}\cmidrule(lr){7-11}
$\eta$ & min & 10\% & median & 90\% & max & min & 10\% & median & 90\% & max \\
\midrule
$-0.50$ & $-35.6$ & $-18.2$ & $-9.2$ & $-3.8$ & $-0.4$ & $-47.8$ & $-46.8$ & $-35.2$ & $-12.0$ & $-2.4$ \\
$-0.25$ & $-15.6$ & $-6.9$ & $-3.3$ & $-1.3$ & $-0.1$ & $-23.5$ & $-22.8$ & $-15.4$ & $-4.4$ & $-0.8$ \\
$-0.10$ & $-5.8$ & $-2.4$ & $-1.1$ & $-0.4$ & $-0.0$ & $-9.3$ & $-9.0$ & $-5.8$ & $-1.5$ & $-0.3$ \\
$+0.10$ & $0.0$ & $0.4$ & $0.9$ & $2.1$ & $5.3$ & $0.2$ & $1.3$ & $5.3$ & $8.8$ & $9.2$ \\
$+0.25$ & $0.1$ & $0.8$ & $2.1$ & $4.6$ & $12.4$ & $0.5$ & $2.8$ & $12.4$ & $21.7$ & $22.8$ \\
$+0.50$ & $0.1$ & $1.3$ & $3.5$ & $8.0$ & $22.6$ & $0.9$ & $4.8$ & $22.6$ & $42.5$ & $44.9$ \\
\bottomrule
\end{tabular}
\end{table}

\bibliographystyle{apalike}
\bibliography{sae_dp}

\end{document}